\documentclass[a4paper,draft=true,DIV=classic,fontsize=10pt]{scrartcl}
\KOMAoptions{DIV=15}

\usepackage[english]{babel}
\usepackage[left=3.5cm,top=3cm,right=3.5cm,bottom=5cm]{geometry}

\usepackage{amsmath,amssymb,amsthm}
\usepackage{enumerate}
\usepackage{graphicx}
\usepackage{subfigure}
\usepackage{arydshln}

\newtheorem{theorem}{Theorem}[]

\newtheorem{corollary}[theorem]{Corollary}
\newtheorem{assumption}{Assumption}
\newtheorem{proposition}{Proposition}

\newtheorem{defintion}{Defintion}[]

\newtheorem*{LNassumption}{Leverage Neutrality Assumption}
\newtheorem*{Hassumption}{$H$-Assumption}

\begin{document}
\title{Theoretical and empirical analysis of trading activity\thanks{The authors acknowledge support by the Vienna Science and Technologie Fund (WWTF) through project MA14-008. M. Pohl and W. Schachermayer are furthermore supported by the Austrian Science Fund (FWF) under the grants P25815 and P28861.
W. Schachermayer additionally appreciates support by the WWTF project MA16-021.}}
\author{Mathias Pohl\thanks{University of Vienna, Faculty of Business, Economics \& Statistics, Oskar-Morgenstern-Platz 1, 1090 Vienna, Austria, mathias.pohl@univie.ac.at} \and Alexander Ristig\thanks{University of Vienna, Faculty of Mathematics and Faculty of Business, Economics \& Statistics, Oskar-Morgenstern-Platz 1, 1090 Vienna, Austria, alexander.ristig@univie.ac.at} \and Walter Schachermayer\thanks{University of Vienna, Faculty of Mathematics, Oskar-Morgenstern-Platz 1, 1090 Vienna, Austria, walter.schachermayer@univie.ac.at} \and Ludovic Tangpi\thanks{Princeton University, Department of Operations Research and Financial Engineering, Sherrerd Hall 203, NJ 08544 Princeton, United States of America, ludovic.tangpi@princeton.edu}}
\date{\today}

\maketitle
\vspace{-1cm}
\begin{center}
Dedicated to Georg Pflug
\end{center}
 \begin{abstract}

  \noindent 
 {\textbf{Abstract}.}
Understanding the structure of financial markets deals with suitably determining the functional relation between financial variables. In this respect, important variables are the trading activity, defined here as the number of trades $N$, the traded volume $V$, the asset price $P$, the squared volatility $\sigma^2$, the bid-ask spread $S$ and the cost of trading $C$. Different reasonings result in simple proportionality relations (``scaling laws'') between these variables. A basic proportionality is established between the trading activity and the squared volatility, i.e., $N \sim \sigma^2$. More sophisticated relations are the so called 3/2-law $N^{3/2} \sim \sigma P V /C$ and the intriguing scaling $N \sim (\sigma P/S)^2$. We prove that these ``scaling laws'' are the only possible relations for considered sets of variables by means of a well-known argument from physics: dimensional analysis. Moreover, we provide empirical evidence based on data from the NASDAQ stock exchange showing that the sophisticated relations hold with a certain degree of universality. Finally, we discuss the time scaling of the volatility $\sigma$, which turns out to be more subtle than one might naively expect.

\end{abstract}

\section{Introduction}

Understanding the structure of financial markets is of obvious relevance for traders, investors and regulators. Among others, the relation between trading activity and price variability received a lot of attention in the financial literature over the last five decades. The  pioneers of this field, e.g.~Clark~\cite{clark1973subordinated}, Epps and Epps~\cite{epps1976stochastic} and Tauchen and Pitts~\cite{tauchen1983price}, defined trading activity via trading volume and derived a proportionality relation between the trading volume and the price variability. The rationale behind this definition and the implied relation is the widely-cited aphorism, ``it takes volume to move prices''. We refer to Karpoff~\cite{karpoff1987relation} for a survey of these early works on the \emph{price-volume relation}. 

Due to minor empirical evidence for the hypotheses developed in these early approaches, the volume-based definition of trading activity has been replaced by the number of trades. This definition is caused by a substantial link between the observed price variability and the number of trades (see Jones et al.~\cite{jones1994transactions}, An\'{e} and Geman~\cite{ane2000order} as well as Dufour and Engle~\cite{dufour2000time}). For example, Jones et al.~\cite{jones1994transactions} find no predictive power in the volume for the price variability but that the number of trades scales proportionally to the squared volatility. This scaling relation will be the starting point of our discussion. Building on the aforementioned ideas numerous other studies followed, e.g.~\cite{andersen1996return,liesenfeld2001generalized}. In particular, let us point out the contribution by Wyart et al.~\cite{wyart2008relation}, who argue that the price volatility per trade, i.e., (price) $\times$ (volatility) $\times$ (number of trades)$^{-1/2}$, is proportional to the bid-ask-spread. This connection can be seen as a somewhat refined version of the relation proposed by Jones et al.~\cite{jones1994transactions}.

More recently, general relations between financial quantities have been derived based on the invariance of markets' microstructure, see Kyle and Obizhaeva~\cite{kyle2016marketempirical}. In particular, the authors postulate a \emph{trading invariance principle} which (in contrast to the above relations) is formulated on the latent level of \emph{meta-orders}.\footnote{A meta-order, also referred to as \emph{bet}, is a collection of trades originating from the same trading decision of a single investor.} Andersen et al.~\cite{andersen2016intraday} and Benzaquen et al.~\cite{benzaquen2016unravelling} confirm empirically that an analogue of this invariance principle holds true for intradaily observable quantities. The fundamental relation may then be formulated  as follows: the nominal value of the exchanged risk during a period of time, defined as the product  (volatility) $\times$ (traded volume) $\times$ (price), is proportional to the number of trades to the power  $3/2$. This so called \emph{intraday trading invariance principle} and its connection to the relations proposed by Jones et al.~\cite{jones1994transactions} and Wyart et al.~\cite{wyart2008relation} is the focus of the present paper. 
\vspace{5mm}

Our aim is to critically analyze these three relations as well as variants thereof by applying a method well known from physics: dimensional analysis. It is a tool which allows for the \emph{falsification} of a proposed relation, e.g.~of the above mentioned formulas for the number of trades, but not for its \emph{verification}. This principle is similar in spirit to K. Popper's approach to epistemology which in turn is inspired by the classical theory of statistics: There one can possibly reject a null hypothesis, but never prove it.
Similarly, dimensional analysis can only isolate those functional relations between variables involving certain ``dimensions'' which do not violate the obvious scaling invariance of these dimensions.
Hence, it a priori rules out those functional relations which are in conflict with these scaling requirements.
But this does \emph{not} imply that the identified functional relations, which are in accordance with the scaling requirements, describe the reality in a reasonable way.
This has to be confirmed by other methods.
In the present setting the ultimate challenge is, of course, to fit to empirical data. 
To complete the picture, we perform an empirical analysis of the relations described above and show that the \emph{intraday trading invariance principle} provides an appropriate fit to empirical data, but fails to be a ``universal law''.

In dimensional analysis one uses the rather obvious argument that a meaningful relation between quantities involving some ``dimensions'' should not be affected by the units in which these ``dimensions'' are measured.
In the present context the relevant ``dimensions'' are time, shares, and money, denoted as $\mathbb{T}, \mathbb{S}$ and $\mathbb{U}$, respectively.
We shall also use an additional argument, namely ``leverage neutrality'' as introduced by Kyle and Obizhaeva~\cite{kyle2017dimensional}. 
We emphasize that these authors were the first to combine the concepts of ``leverage neutrality'' and  dimensional analysis.
The assumption of leverage neutrality is based on the Modigliani-Miller theorem (see~\cite{modigliani1958cost}) and leads to a scaling invariance principle which, mathematically speaking, is perfectly analogous to the dimensional scaling requirements mentioned above. 

The remainder of the paper is structured as follows. In Section \ref{sec:TradInv}, we first deduce the proportionality between the number of trades and the price variability as proposed by Jones et al.~\cite{jones1994transactions} from dimensional arguments. Next, we derive the more involved scaling relations proposed by Benzaquen et al.~\cite{benzaquen2016unravelling} as well as Wyart et al.~\cite{wyart2008relation}, again using dimensional analysis, and discuss the assumption of leverage neutrality in this context. Having a theoretical foundation for the discussed relations, we then turn to the empirical analysis in Section \ref{sec:EmpEvi}: Based on data from the NASDAQ stock market, we show that the relation proposed by Benzaquen et al.~\cite{benzaquen2016unravelling} fits the data rather well. 
In Section \ref{sec:VolaScaling}, we take a closer look at
volatility and analyze implications of different time scalings thereof.
We conclude with some empirical results in this respect. 
A reminder on the Pi-theorem from dimensional analysis as well as proofs for all considered relations can be found in the appendix.

\section{The trading invariance principle}
\label{sec:TradInv}
We are interested in explaining the arrival rate of trades in a given stock measured as
\begin{itemize}
\item $N = N_t^{t+T}\quad\,\,$ the number of trades within a fixed time interval $[t,t+T]$ so that $N$ is measured per units of time. Following the notation from \cite{pohl2017amazing}, this link between the variable $N$ and its dimensional unit is therefore given by $$[N] = \mathbb{T}^{-1}.$$
\end{itemize}  
Let us identify the variables (and their dimensions $[\cdot]$)  which are likely to influence the number of trades $N$ in a given interval $[t,t+T]$.
Three obvious candidates are:
\begin{itemize}
	\item $V = V_t^{t+T}\quad\,\,$ the traded volume of the stock during the time interval $[t,t+T]$, measured in units of shares per time $$[V] = \mathbb{S}/\mathbb{T}.$$
	\item $P = P_t^{t+T}\quad\,\,$ the average price of the stock in the interval $[t,t+T]$, measured in units of money per share $$[P] = \mathbb{U}/\mathbb{S}.$$
	\item $\sigma^2= (\sigma^2)_t^{t+T}=\mathbb{V}\text{ar} \left(\log(P_{t+T}) - \log(P_t)\right)\quad $ the variance of the log-price over the time interval $[t,t+T]$.
	We assume $$[\sigma^2] = \mathbb{T}^{-1}.$$
\end{itemize}
If the price process $(P_t)_{t\geq 0}$ follows, e.g.~the Black-Scholes model, see~\eqref{eq:Black-Scholes}, we clearly find the above scaling $[\sigma^2]=\mathbb{T}^{-1}$ and shall retain this  assumption in most of the paper. However, the  scaling of $\sigma^2$ turns out to be more subtle than it seems at first glance. In Section \ref{sec:VolaScaling} below, we shall investigate the implications of a scaling relation $[\sigma^2] = \mathbb{T}^{-2H},$ where $H \in (0,1)$ may be different from $1/2$. For instance, such a scaling may result from price processes based on a fractional Brownian motion $(B^H_t)_{t\geq 0}$ with Hurst parameter $H \in (0,1)$, see~\cite{mandelbrot1968fractional}.
\vspace{5mm} 

Based on these identified dimensions, let us turn to the basic idea of dimensional analysis: the validity of a considered relation should not depend on whether we measure time $\mathbb{T}$ in seconds or in minutes, shares $\mathbb{S}$ in single shares or in packages of hundred shares, 
and money $\mathbb{U}$ in Euros or in Euro-cents. 
\begin{defintion}[Dimensional invariance]
A function $h:\mathbb{R}^n_+ \rightarrow \mathbb{R}_+$ relating the quantity of interest $U$ to the explanatory variables $W_1,\dots,W_n$, i.e, $$U=h(W_1,\dots,W_n),$$ is called \emph{dimensionally invariant} if it is invariant under rescaling  the involved dimensions (in our case $\mathbb{S}, \mathbb{T}$ and  $\mathbb{U}$).
\end{defintion}

As a first - and rather naive - approach we analyze the assumption that the three variables $\sigma^2,P$ and $V$ \emph{fully} explain the number of trades $N$. 
\begin{proposition}
\label{pro:naive}
Assume that the number of trades $N$ depends \emph{only} on the three quantities $\sigma^2,P$ and $V$, i.e.,
\begin{align}
\label{eq:naive}
	N &= g(\sigma^2,P,V),
\end{align}
 where the function $g:\mathbb{R}_+^3\rightarrow\mathbb{R}_+$ is \emph{dimensionally invariant}. Then, there is a constant $c>0$ such that the number of trades $N$ obeys the relation
	\begin{align} \label{eq:N=sigma^2}
		N = c \cdot \sigma^2.
	\end{align}
\end{proposition}
The proof relies on elementary linear algebra and is given in Appendix \ref{sec:Proofs} below (compare also the proof of Theorem 1 below which is similar).
Recall that relation \eqref{eq:N=sigma^2} goes back to Jones et al.~\cite{jones1994transactions}.

As mentioned in the introduction, one should read the present ``dimensional'' argument in favor of relation \eqref{eq:N=sigma^2} as a pure ``if$\dots$then$\dots$'' assertion: \textbf{if} $N$ really is fully explained by $\sigma^2,P$ and $V$ \textbf{and} the obvious scaling invariances of $\mathbb{S}$, $\mathbb{T}$ and $\mathbb{U}$ are satisfied, \textbf{then} \eqref{eq:N=sigma^2} is the only possible relation. As we shall see below, the empirical data does not reconfirm the validity of \eqref{eq:N=sigma^2}. In other words, we have to turn the above statement upside down: as \eqref{eq:N=sigma^2} is not reconfirmed by empirical data, the variables $\sigma^2,P$ and $V$ cannot fully explain the quantity $N$. It is therefore natural to introduce more/other quantities in order to explain the number of trades $N$. 

Regarding the uniqueness of the function $g$ in \eqref{eq:naive}, the mathematical reason for the unique choice of $g$ given by 
\eqref{eq:N=sigma^2} is that we have three scaling relations (pertaining to the invariance of the ``dimensions'' $\mathbb{S}, \mathbb{U}$ and $\mathbb{T}$) as well as the three explanatory variables $\sigma^2, P$ and $V$. This leads to three linear equations in three unknowns, yielding a unique solution. 
\vspace{5mm}

Let us now try to go beyond the scope of relation \eqref{eq:naive} by considering further explanatory variables. Motivated by Wyart et al.~\cite{wyart2008relation}, we consider the following quantity as relevant for the number of trades $N$ in a given interval $[t,t+T]$, additionally to $\sigma^2, P$ and $V$: 
\begin{itemize}
\item $S = S_t^{t+T}\quad\,\,$ the average bid-ask spread in the interval $[t,t+T]$, measured in units of money per share $$[S] = \mathbb{U}/\mathbb{S}.$$
\end{itemize}
Following Benzaquen et al.~\cite{benzaquen2016unravelling}, it is also convenient to alternatively consider the quantity
\begin{itemize}
\item $C=C_t^{t+T}\quad\,\,$ the average cost per trade in the interval $[t,t+T]$, measured in units of money $$[C] = \mathbb{U}.$$
\end{itemize}

To visualize things, suppose that for some stock we observe in average during the time interval $[t,t+T]$ an ask price of EUR$12.30$ and a bid price of EUR$12.20$ so that the bid-ask spread $S$ equals 10 cents. If the average trade size in the interval $[t,t+T]$, denoted by $Q=Q_t^{t+T}$, is 500 shares, we obtain that the average cost per trade $C = QS$ is EUR$50$. 
A discussion of the difference between using $S$ rather than $C$ as an explanatory variable can be found at the end of this section. For now, let us follow Benzaquen et al.~\cite{benzaquen2016unravelling}  for our derivation of the \emph{intraday trading invariance principle} and pass to the set $\sigma^2, P,V$ and $C$ of explanatory variables, i.e.,
\begin{align}\label{eq:N=g(sigma^2,P,V,C)1}
	N &= g(\sigma^2,P,V,C),
\end{align}
for some function $g:\mathbb{R}_+^4\rightarrow\mathbb{R}_+$. As we now have four explanatory variables, the three equations yielded by the scale invariance of the dimensions $\mathbb{S}, \mathbb{U}$ and $\mathbb{T}$ are not sufficient anymore to imply an (essentially) unique solution for $g$. In fact, the four explanatory variables above combined with the three invariance relations pertaining to  $\mathbb{S}$, $\mathbb{T}$ and  $\mathbb{U}$ only yield a general solution of \eqref{eq:N=g(sigma^2,P,V,C)1} of the form
\begin{align}
\label{eq:UnknownFunctionf}
N= \sigma^2 f\left( \frac{PV}{\sigma^2 C} \right),
\end{align}
where $f:\mathbb{R}_+ \rightarrow \mathbb{R}_+$ is an arbitrary function whose generality cannot be restricted by only relying on arguments pertaining to dimensional analysis with respect to the three dimensions $\mathbb{S}$, $\mathbb{T}$ and  $\mathbb{U}$ (see Appendix \ref{sec:Proofs}).
\vspace{5mm}

Hence, in order to obtain such a crisp result as in \eqref{eq:N=sigma^2}, an additional ``dimensional invariance'' is required. Kyle and Obizhaeva~\cite{kyle2017dimensional} found a remedy: a no-arbitrage type argument, referred to as ``leverage neutrality''.\footnote{Note that Kyle and Obizhaeva~\cite{kyle2017dimensional} use the argument of leverage neutrality in the context of market impact. But, of course, the same idea applies in the present situation.}
This concept is inspired by the findings of Modigliani and Miller~\cite{modigliani1958cost} (compare~\cite{pohl2017amazing}): Consider a stock of a company, and suppose that the company changes its capital structure by paying dividends or by raising new capital. The Modigliani-Miller theorem tells us precisely which features of the company are \emph{not affected} by a change in the capital structure. This allows us to establish how certain quantities behave when varying the leverage in terms of the relation between debt and equity of a company.

From a conceptual point of view, the assumption of leverage neutrality gives a constraint on the behavior of the quantities $N, \sigma^2, P, V, C$ (resp. $S$) in case of changing the firm's capital structure.
This constraint can be understood as an additional though synthetic dimension in our analysis, which we refer to as the Modigliani-Miller ``dimension'' $\mathbb{M}$. The Modigliani-Miller ``dimension'' $\mathbb{M}$ of a share of a company is measured in terms of the leverage $\mathcal{L}$, i.e., the quantity 
\begin{align*}
	\mathcal{L} = \frac{\text{total assets}}{\text{equity}}.
\end{align*}
Multiplying $\mathcal{L}$ by a factor $A > 1$  is equivalent to paying out $(1-A^{-1})$ of the equity as cash-dividends. On the other hand, multiplying $\mathcal{L}$ by a factor $0<A<1$ corresponds to raising new capital in order to increase the firm's equity by a factor $A^{-1}$. Following Kyle and Obizhaeva~\cite{kyle2017dimensional} as well as~\cite{pohl2017amazing}, we are led to the following assumption:

\begin{LNassumption}[\cite{kyle2017dimensional,pohl2017amazing}]
Scaling the Modigliani-Miller ``dimension'' $\mathbb{M}$ by a factor $A \in \mathbb{R}_+$ implies that 
\begin{itemize}
\item $N$, $V$ and $C$ (as well as $S$) remain constant,
\item $P$ changes by a factor $A^{-1}$,
\item $\sigma^2$ changes by a factor $A^2$.
\end{itemize}
\end{LNassumption}

To recapitulate: Setting $A=2$ corresponds to paying out half of the equity as dividends so that each share yields a dividend of $(1-A^{-1})P=P/2$. The stock price is, thus, multiplied by $A^{-1}=1/2$ while the volatility $\sigma$ is  multiplied by $A=2$.
The remaining quantities are not affected by changing the leverage, in accordance with the insight of Modigliani and Miller~\cite{modigliani1958cost} and the recent work by  Kyle and Obizhaeva~\cite{kyle2017dimensional}.  The economic reason is that the value of the assets of the corresponding company and hence the associated risk does not change. 
\begin{defintion}[Leverage neutrality]
A function $h:\mathbb{R}_+^n \rightarrow \mathbb{R}_+$ relating the quantity $N$ to the explanatory variables $\sigma^2, P, V, C$ and $S$, i.e, $$N=h(\sigma^2, P,V, C, S),$$ is called \emph{leverage neutral} if it is invariant when rescaling the Modigliani-Miller dimension $\mathbb{M}$ of the variables $N, \sigma^2, P, V, C, S$ as defined in the assumption above.
\end{defintion}

We can now derive the following relation, which is the focus of the present paper. It relies on the basic fact that under the ``Leverage Neutrality Assumption'' we now find four linear equations in order to determine four unknowns.
Note that Benzaquen et al.~\cite{benzaquen2016unravelling} coined this relation the ``3/2-law''. 
\begin{theorem}[$(3/2)$-law] 
\label{thm:3/2law}
Suppose the ``Leverage Neutrality Assumption'' holds and that the number of trades $N$ depends \emph{only} on the four quantities $\sigma^2, P, V$ and $C$, i.e.,
\begin{align}
\label{eq:N=g(sigma^2,P,V,C)THM}
	N &= g(\sigma^2,P,V,C),
\end{align}
 where the function $g:\mathbb{R}_+^4\rightarrow\mathbb{R}_+$ is \emph{dimensionally invariant} and \emph{leverage neutral}. Then, there is a constant $c>0$ such that the number of trades $N$ obeys the relation
\begin{align}\label{eq:3/2law}
	N^{3/2} =c\, \cdot\, \frac{\sigma PV}{C}.
\end{align}
\end{theorem}

The proof follows from the general Pi-theorem reviewed in Appendix \ref{sec:DAandPiTheorem}.
For the convenience of the reader, we also present a direct proof of Theorem \ref{thm:3/2law}.
Although slightly longish and repetitive, we hope that it helps the intuition.
\begin{proof}[Proof of Theorem \ref{thm:3/2law}]
	First, we make the following \emph{ansatz} for the function $g$ in \eqref{eq:N=g(sigma^2,P,V,C)THM}:
	\begin{equation}
	\label{eq:E1a}
		g(\sigma^2,P,V,C) = c\cdot (\sigma^2)^{y_1}P^{y_2}V^{y_3}C^{y_4},
	\end{equation}
	where $c>0$ is a constant and $y_1,\dots,y_4$ are unknown real numbers.
	Looking at the first row of Table \ref{tab:Dimensions} yields the relation 
	\begin{equation}
	\label{eq:E1}
		-y_2 + y_3 = 0.
	\end{equation}
	Indeed, when passing from counting shares in packages of $100$ units rather than in single units, the number $P$ is replaced by $100P$ while the number $V$ is replaced by $V/100$. Since the function $g$ in \eqref{eq:E1a} is assumed to be dimensionally invariant, $g$ should remain unchanged by this passage, i.e.,
	\begin{equation}
	\label{eq:E2a}
		c\cdot \left(\sigma^2\right)^{y_1}P^{y_2}V^{y_3}C^{y_4} = c\cdot\left(\sigma^2\right)^{y_1}\left(100P\right)^{y_2}\left(\frac{V}{100}\right)^{y_3}C^{y_4}
	\end{equation}
	which is only possible if \eqref{eq:E1} holds true.
	Looking at the other rows of Table \ref{tab:Dimensions} we therefore get the system of linear equations 
	\begin{equation*}
		\begin{cases}
			\qquad - \,\,y_2 + y_3&=\,\,\,\,0\\
			\qquad \quad y_2\qquad\,\, + y_4 &= \,\,\,\,0\\
			-y_1 \qquad\,\,-  y_3&=-1\\
			\,\,2y_1- y_2&=\,\,\,\,0
		\end{cases}
	\end{equation*}
	whose unique solution is \begin{align} \label{eq:solutiony}
	y = \left( \frac{1}{3}, \frac{2}{3}, \frac{2}{3},  -\frac{2}{3} \right)^\top,
\end{align}	
	which gives \eqref{eq:3/2law} as one possible solution of \eqref{eq:N=g(sigma^2,P,V,C)THM}.

	We still have to show the uniqueness of \eqref{eq:3/2law}.
	To do so, it is convenient to pass to logarithmic coordinates:
	suppose that there is a function $G:\mathbb{R}^4 \to \mathbb{R}$ such that $\log(N) = G\left( \log(\sigma^2),\log(P), \log(V),\log(C)\right)$ or equivalently,
	\begin{equation}
	\label{eq:E2}
		\log(N) - G(X_1, X_2, X_3, X_4) =0,
	\end{equation}
	where we write $\left( \log(\sigma^2),\log(P), \log(V), \log(C)\right)$ as $(X_1, X_2, X_3, X_4)$.
	We have to show that $G$ has the form 
	\begin{equation*}
		\log(N) = y_1X_1 + y_2 X_2 + y_3X_3 +y_4X_4 + \text{const},
	\end{equation*}
	where $y_1, y_2, y_3, y_4$ are given by \eqref{eq:solutiony} and const is a real number.
	Denote by $r_1:= -e_2 + e_3$ the first row of Table \ref{tab:Dimensions}, considered as a vector in $\mathbb{R}^4$, where $(e_i)_{i=1}^{4}$ is the canonical basis of $\mathbb{R}^4$.
	Similarly as in \eqref{eq:E2a}, the first row of Table \ref{tab:Dimensions} and dimensional invariance imply that
	\begin{align*}
		G&\left(\log(\sigma^2),\log(P), \log(V), \log(C)\right) \\&= G\left(\log(\sigma^2),\log(P)+ \log(100), \log(V) - \log(100), \log(C)\right).
	\end{align*}
	Clearly we can replace $\log(100)$ by any real number.
	Speaking abstractly, this means that $G:\mathbb{R}^4\to \mathbb{R}$ must be constant on any straight line parallel to the vector $r_1$.
	A similar argument applies to $r_2 = e_2 + e_4$ and $r_4=2e_1 - e_2$.
	As regard $r_3= -e_1-e_3$ the situation is slightly different, as the third row of Table \ref{tab:Dimensions} also involves a non-zero entry of $N$.

	The third row of Table \ref{tab:Dimensions} and \eqref{eq:E2} imply that for any $\lambda \in \mathbb{R}$,
	\begin{equation*}
		G(X_1-\lambda, X_2 , X_3 - \lambda, X_4) = G(X_1, X_2, X_3, X_4) - \lambda.
	\end{equation*}
	Setting const $:= G(0,0,0,0)$, we have 
	\begin{equation*}
		G(-\lambda, 0, -\lambda,0) = - \lambda + \text{const} \quad \text{for all } \lambda \in \mathbb{R},
	\end{equation*}
	which uniquely determines $G$ on the one-dimensional space spanned by $r_3 = -e_1-e_3$ in $\mathbb{R}^4$.
	As we have seen that $G$ also must be constant along each line in $\mathbb{R}^3$ parallel to $r_1, r_2$ and $ r_4$, and as $r_1, r_2, r_3, r_4$ span the entire space $\mathbb{R}^4$, we conclude that there is only one choice for the function $G$, up to the constant $\text{const} = G(0,0,0,0)$.
\end{proof} 

\begin{table}
\begin{center}
\begin{tabular}{c|cccc|c}
		&$\sigma^2$ & $P$ & $V$   & $C$ & $N$\\ 
		\hline
		$\mathbb{S}$ &0  &-1 & 1 & 0 & 0\\
		$\mathbb{U}$ & 0 & 1 & 0 & 1 & 0\\
		$\mathbb{T}$ &-1 & 0 &-1 & 0 & -1\\
		\hdashline
		$\mathbb{M}$ &2  & -1 &0 & 0 & 0\\
\end{tabular}
\end{center}
\caption{A labelled overview of the dimensions of the quantities $P, V, \sigma^2$ and $C$.}
		\label{tab:Dimensions}
\end{table}

For an alternative derivation of relation \eqref{eq:3/2law}, we pass from considering $\sigma^2$, the variability of the \emph{relative} price changes, to considering $\sigma_B^2$, the variability of the \emph{absolute} price changes.
This will allow us to reduce the \emph{two} explanatory variables $\sigma^2$ and $P$ to \emph{one} explanatory variable $\sigma_B^2 = \sigma^2 P^2$. 
We call $\sigma_B$ the \emph{Bachelier volatility} as it corresponds to Bachelier's original model from 1900, see~\cite{bachelier1900theorie}.
Recall that the dynamics of the price process $(P_t)_{t\ge 0}$ of the Black-Scholes versus the Bachelier model are
\begin{align}
	\label{eq:black-scholes}dP_t &= \sigma P_t dW_t, \quad &\text{(Black-Schloes model)} \\
	dP_t &= \sigma_B dW_t,\quad &\text{(Bachelier model)}\notag
\end{align}
where $W_t$ is a standard Brownian motion. Defining $\sigma_B=\sigma P$ the two models coincide remarkably well as long as $P_t$ does not move too much (compare e.g.~\cite{Schac-Teich08}).
We therefore define 
\begin{itemize}
	\item $\sigma_B^2 =\sigma^2 P^2$ the Bachelier volatility in the interval $[t,t+T]$. 
	Plugging in the dimensions $[\sigma^2]=\mathbb{T}^{-1}$ and $[P] = \mathbb{U}\mathbb{S}^{-1}$, we obtain
  $$[\sigma_B^2] = \mathbb{U}^2 \mathbb{S}^{-2} \mathbb{T}^{-1}.$$
\end{itemize}
A glance at Table \ref{tab:DimensionsBachelier} reveals that $\sigma^2_B$ has Modigliani-Miller dimension $\mathbb{M}$ equal to zero (just as the other variables $V, C$ and $N$).
This enables us to derive the assertion of Theorem \ref{thm:3/2law} by using only the three obvious scaling invariances, but \emph{without} imposing a priori the requirement of leverage neutrality.
\begin{corollary} \label{coro:Corollary3/2Bachelier}
	Suppose the number of trades $N$ depends \emph{only} on the three quantities $\sigma_B^2,V$ and $C$, i.e.,
	\begin{align}\label{eq:N=g(sigma_B^2,V,C)}
		N &= g(\sigma_B^2,V,C),
	\end{align}
	where the function $g:\mathbb{R}_+^3\rightarrow\mathbb{R}_+$ is \emph{dimensionally invariant}. Then, there is a constant $c>0$ such that the number of trades $N$ obeys the relation
	\begin{align} \label{eq:3/2BachelierVersion}
		N^{3/2}=c\, \cdot\, \frac{\sigma_B V}{C}.
	\end{align}
\end{corollary}

The proof is analogous to (and even easier than) the above proof. 
Note that Proposition \ref{pro:naive} and Corollary \ref{coro:Corollary3/2Bachelier} both only rely on the very convincing invariance assumption with respect to  $\mathbb{S}$, $\mathbb{T}$ and $\mathbb{U}$, but not on the ``Leverage Neutrality Assumption''.

Anticipating that relation \eqref{eq:3/2BachelierVersion} gives a superior fit to empirical data than relation \eqref{eq:N=sigma^2} we can draw the following conclusion: the choice of $\sigma_B^2, V, C$ as explanatory variables for the quantity $N$ is superior to the choice $\sigma^2, P, V$ made in Proposition \ref{pro:naive} above.

Here is a ``dimensional argument'' why we should expect a better result from Corollary \ref{coro:Corollary3/2Bachelier} as compared to Proposition \ref{pro:naive}.
It follows from the very approach of dimensional analysis that everything hinges on the assumption that the chosen explanatory variables indeed ``fully explain'' the dependent variable.
Of course, in reality such an assumption will -- at best -- only be approximately satisfied.
The art of the game is to find a combination of explanatory variables which ``best'' explain the resulting variable.
The choice of the variables $\sigma^2_B, V, C$ as in Corollary \ref{coro:Corollary3/2Bachelier} \emph{automatically} implies that the ``Leverage Neutrality Assumption'' is satisfied as shown in Table \ref{tab:DimensionsBachelier}.
Indeed, the variables $\sigma^2_B, V, C$ as well as $N$ have a zero entry for the Modigliani-Miller dimension $\mathbb{M}$.
Therefore, \emph{any} function relating these variables is \emph{automatically} leverage neutral.
This is in contrast to the choice of variables $\sigma^2, P, V$ in Proposition \ref{pro:naive} as
Table \ref{tab:Dimensions} reveals that $P$ and $\sigma^2$ have a non-trivial dependence on $\mathbb{M}$.
It follows that formula \eqref{eq:N=sigma^2} does not satisfy the invariance relation dictated by the ``Leverage Neutrality Assumption''.
\vspace{5mm}

\begin{table}
\begin{center}
\begin{tabular}{c|ccc|c}
		 & $\sigma_B^2$& $V$  & $C$ & $N$\\ 
		\hline
		$\mathbb{S}$  & -2 & \,\,1 & 0 & \,\,0\\
		$\mathbb{U}$  &  \,\,2 & \,\,0 & 1 & \,\,0\\
		$\mathbb{T}$  & -1 &-1& 0 & -1\\
		\hdashline
		$\mathbb{M}$  & 0 & 0  & 0 & 0\\
\end{tabular}
\end{center}
\caption{A labelled overview of the dimensions of the quantities $V, \sigma_B^2 = \sigma^2 P^2$ and $C$.}
		\label{tab:DimensionsBachelier}
\end{table}

Finally, we examine the implications of substituting the cost per trade $C$ by its more common counterpart, the bid-ask spread $S$, introduced above. In fact, in the present context it is equivalent to use either $C$ or $S$ as explanatory variables for the number of trades $N$ - provided that the traded volume $V$ is already one of the explanatory variables. Indeed, we have the relation $C = SQ = SV/N$ since the average trade size $Q$ in the interval $[t,t+T]$ is given by the traded volume $V$ divided by the number of trades $N$.
Hence, if we know the functional relation between $N$ and $V$, we also know the functional relation between $N$ and $Q$ and can therefore pass from $S$ to $C=SQ$ and vice versa. 
Thus, we may restate Theorem \ref{thm:3/2law} (and, equivalently, Corollary \ref{coro:Corollary3/2Bachelier}) in terms of the bid-ask spread $S$ rather than the cost per trade $C$ in the following corollary. 
\begin{corollary}
\label{cor:sigma2OverS}
Suppose that the number of trades $N$ depends \emph{only} on the three quantities $\sigma_B^2$, $V$ and $S$, i.e.,
\begin{align}\label{eq:N=g(sigma_B^2,V,S)}
	N &= g(\sigma_B^2,V,S),
\end{align}
where the function $g:\mathbb{R}_+^3\rightarrow\mathbb{R}_+$ \emph{dimensionally invariant}. Then, there is a constant $c>0$ such that the number of trades $N$ obeys the relation
\begin{align} \label{eq:N=sigmaB/S^2}
	N=c^2\, \cdot\, \left( \frac{\sigma_B}{S} \right)^{2}.
\end{align}
\end{corollary}

We observe that the variables $\sigma^2_B, V$ and $S$ again have no Modigliani-Miller dimension $\mathbb{M}$, i.e., they are invariant under changes of the leverage.
Therefore, formula \eqref{eq:N=sigmaB/S^2} satisfies the invariance principle given by the ``Leverage Neutrality Assumption''. We note again that given the relations $C = SQ = SV/N$ as well as $\sigma_B^2 = \sigma^2 P^2$ the two equations \eqref{eq:3/2law} and \eqref{eq:N=sigmaB/S^2} are indeed equivalent. 

Relation \eqref{eq:N=sigmaB/S^2} is precisely the one proposed by Wyart et al.~\cite{wyart2008relation}. By rearranging the terms, we find that
\begin{align}
\label{eq:S=sigma/N}
S^2 = c^2 \cdot  \frac{\sigma_B^2}{N}.
\end{align}
The interpretation is that the squared Bachelier volatility per trade is proportional to the square of the spread. If we elaborate further on \eqref{eq:S=sigma/N}, we find that
\begin{align}\label{eq:S=sigma/N2}
\frac{S}{P} = c \cdot  \frac{\sigma}{\sqrt{N}}.
\end{align}
Without loss of generality, we can determine the price $P$ on the left hand side of \eqref{eq:S=sigma/N2} as midquote price, i.e., the average of the best ask- and bid price. Then, $S/P$ refers to the so called proportional bid-ask spread which can be used to approximate a dealer's ``round trip'' transaction costs. Clearly, the approximate round-trip costs increase in the volatility of a relative price change and decrease in the trading activity.
\vspace{5mm}

Summing up this section, we have seen that the relation $N \sim \sigma^2$ proposed by Jones et al.~\cite{jones1994transactions} follows from the restrictive assumption that the number of trades $N$ \emph{only} depends on the quantities $\sigma^2, P$ and $V$ as well as dimensional arguments (see Proposition \ref{pro:naive}). Going beyond the latter relation, it seems reasonable to include information concerning the bid-ask spread in our analysis. Depending on whether we choose the trading cost $C$ or the bid-ask spread $S$ directly, we are led to either the 3/2-law $N^{3/2} \sim \sigma P V /C$ proposed by Benzaquen et al.~\cite{benzaquen2016unravelling} (see Theorem \ref{thm:3/2law}) or to the relation $S \sim \sigma_B/\sqrt{N}$ proposed by Wyart et al.~\cite{wyart2008relation} (see Corollary \ref{cor:sigma2OverS}). When proving the two latter relations we have seen that the assumption of leverage neutrality comes into play. Alternatively, we can also consider the product $\sigma^2 P^2$, rather than $\sigma^2$ and $P$  separately. This consideration of the ``Bachelier volatility'' $\sigma_B = \sigma P$ reduces the complexity of the problem inasmuch as the assumption of leverage neutrality is not needed anymore.
Again, the \emph{actual} validity of any of the above scaling laws should be confirmed by exhaustive empirical analysis.

\section[Empirical evidence]{Empirical evidence}\label{sec:EmpEvi}

\subsection{Degrees of universality and relevant literature}

We now turn to the empirical analysis of relation \eqref{eq:N=sigma^2} as well as of the $3/2$-law \eqref{eq:3/2law}. When collecting data for the quantities $N$, $\sigma^2$, $V$, $P$ and $C$, one has to specify the considered asset and the considered time period as well as the length $T$ of the time interval  over which the data is aggregated. We cannot expect that the constant $c$ appearing in relations 
\eqref{eq:N=sigma^2} resp. \eqref{eq:3/2law} is the same for each considered interval \emph{and} each possible interval length \emph{and} each considered asset in either one of the relations. We can only hope that a given relation holds \emph{on average}.  
Based on the nomenclature introduced in Benzaquen et al.~\cite{benzaquen2016unravelling}, we therefore distinguish the following three degrees of universality attached to the validity of  relations \eqref{eq:N=sigma^2} and \eqref{eq:3/2law}:
\begin{enumerate}
\item \emph{No universality:} The relation holds on average for a fixed asset and a fixed interval length. However, the constant $c$ varies significantly for different assets and different interval lengths.
\item \emph{Weak universality:} The relation holds on average for some assets and some interval lengths with similar values from the constant $c$.
\item \emph{Strong universality:} The relation holds on average for all assets and all interval lengths with similar values from the constant $c$.
\end{enumerate} 
Note that this distinction does not allow for the possibility that the validity attached to a given relation changes over time, simply because we consider only one specific time period.
\vspace{5mm}

Let us shortly discuss the relevant empirical evidence which can be found in the literature before turning to our own empirical analysis.
Andersen et al.~\cite{andersen2016intraday} conducted an important empirical study in the present context. They test the relation 
\begin{align}
\label{eq:invarianteI}
I = \frac{\sigma P V}{N^{3/2}},
\end{align}
where $I$ is independently
and identically distributed across assets and time for E-mini S\&P 500 futures contract. Neglecting the price $P$, they show that relation $N^{3/2} \sim V \sigma$ holds when averaging within and across trading days for this particular asset. In fact, their data fits the latter relation nearly perfectly compared to the relations $V \sim \sigma^2$ resp. $N\sim \sigma^2$ proposed by Tauchen and Pitts~\cite{tauchen1983price} resp. Jones et al.~\cite{jones1994transactions}.
Benzaquen et al.~\cite{benzaquen2016unravelling} address the same question by examining eleven additional futures contracts as well as 300 US stocks. Aiming to confirm that $\beta = 3/2$ in the relation $N^{\beta} \sim \sigma P V$, they estimate $\beta$ for each considered stock individually. 
They find that $\hat{\beta} = 1.54 \pm 0.11$, where the uncertainty here is the root mean square cross-sectional dispersion. Thus, these authors note that this provides evidence that the relation $N^{3/2} \sim  \sigma P V $ holds also on the stock market and not only on the very liquid futures market. Moreover, they show that the distribution of $I$ in \eqref{eq:invarianteI} depends significantly on the studied asset and thus, conclude that relation \eqref{eq:invarianteI} holds only with weak universality. As an additional contribution, the authors reveal that the inclusion of the trading cost $C$ is beneficial in the sense that their proposed invariant $\mathcal{I} = \sigma P V C^{-1} N^{-3/2}$ is almost constant for different assets. 

Finally, let us mention the evidence in the earlier work by Wyart et al.~\cite{wyart2008relation}.
These authors show that relation \eqref{eq:S=sigma/N} describes the data very well when the right level of aggregation is chosen. When examining the France Telecom stock, $S$ and $\sigma_B/\sqrt{N}$ are averaged over two trading days, while in case of NYSE stocks these quantities are averaged over an entire year. The constant $c$ in relation \eqref{eq:S=sigma/N} is found to lie between $1.2$ and $1.6$. Moreover, the authors note that the typical intraday pattern of the considered quantities is in line with \eqref{eq:S=sigma/N}: The U-shaped pattern of the volatility $\sigma_B$ is explained by the decline of the bid-ask spread $S$ and an increase of the number of trades $N$ within the trading day.

\subsection{Description of data} \label{ssec:Data}
Our empirical analysis is based on limit order book data provided by the {LOBSTER} database (\texttt{https://lobsterdata.com}). The considered sampling period begins on January 2, 2015 and ends on August 31, 2015, leaving 167 trading days. Among all NASDAQ stocks, $d=128$ sufficiently liquid stocks with high market capitalizations are chosen. Stocks are considered to be ``sufficiently liquid'' as long as the aggregated variables (defined below) can be reasonably treated as continuously distributed, i.e., the empirical distributions of the aggregated variables do not have points with obviously concentrated mass. Observations made during the thirty minutes after the opening of the exchange as well as trading halts are removed.
\vspace{5mm}

Let us fix an interval length $T\in\{30, 60, 120, 180, 360\}$ min for which a developed hypothesis is tested. For the sake of illustration, set the length of the considered time interval $T$ to 60min. This interval length balances the tradeoff between sufficient aggregation of the data on the one hand and some intraday variability on the other hand. As a result, we are left with $n=1002$ non-overlapping time intervals with equal length $T=60$min. Let us concentrate on a specific asset $i \in \lbrace 1,\dots,d \rbrace$ (omitting the index $i$ for ease of notation in the remainder of Section \ref{ssec:Data}) and let $j\in \lbrace 1,\dots,n\rbrace $ refer to an arbitrary interval. Suppose the trades in the considered
interval $j$ arrive at irregularly spaced transaction times $t_1,t_2,\ldots,t_{N_j}$.
Then, 
\begin{itemize}
	\item[$N_j$] denotes the number of trades in the interval $j$,
	\item[$Q_j$] $= N_j^{-1} \sum_{k=1}^{N_j} Q_{t_k}$ denotes the average size of the trades in the interval $j$, where $Q_{t_k}$ denotes the number of shares traded at time $t_k$,
	\item[$V_j$] $= N_j \times Q_j$ is the traded volume in the interval $j$,
	\item[$P_j$] $= N_j^{-1} \sum_{k=1}^{N_j} P_{t_k}$ denotes the average midquote price in the interval $j$, where $P_{t_k} = (A_{t_k} + B_{t_k})/2$ and $A_{t_k}$ (resp. $B_{t_k}$) denotes the best ask (resp. bid) price after the transaction at time $t_k$, 
	\item[$\hat{\sigma}_j^2$] denotes the estimated squared volatility in the interval $j$,
	\item[$S_j$] $= N_j^{-1} \sum_{k=1}^{N_j} S_{t_k}$ denotes the average bid-ask spread in the interval $j$, where $S_{t_k} = A_{t_k} - B_{t_k}$ is the bid-ask spread after the transaction at time $t_k$, and 
	\item[$C_j$] $=Q_j\times S_j$ is the cost per trade in the interval $j$.
\end{itemize}

Note the following four details: Firstly, even though transaction times are recorded on a nano-second level, a time-stamp $t_k$ is recorded $L$-times ($t_{k_1}, \dots, t_{k_L}$) in the raw dataset when a market order is executed against $L$ limit orders at time $t_k$. Such a multiple entry of the same time-stamp enters the number of trades $N_j$ only once (not $L$-times). The size $Q_{t_k}$  of the trade at time $t_k$ is determined by summing the $L$-records in the dataset $Q_{t_{{k}_{\ell}}}$, $\ell=1,\ldots,L$, i.e., $Q_{t_k} = \sum_{\ell=1}^L Q_{t_{k_\ell}}$. The midquote price $P_{t_k}$ and the bid-ask spread $S_{t_k}$ related to the merged market order of size $Q_{t_k}$ are computed as volume-weighted averages
\begin{align*}
	P_{t_k}=Q_{t_k}^{-1}\sum_{\ell=1}^LQ_{t_{{k}_{\ell}}}P_{t_{{k}_{\ell}}}\quad\text{and}\quad S_{t_k}=Q_{t_k}^{-1}\sum_{\ell=1}^LQ_{t_{{k}_{\ell}}}S_{t_{{k}_{\ell}}}.
\end{align*}

Secondly, the aggregated variables, i.e., the average market order size $Q_j$, the average midquote price $P_j$ and the average bid-ask spread $S_j$ of interval $j$, are in fact not computed by the sample averages as state above. Since simple sample averages are sensitive with respect to outliers, e.g.~huge market orders, $Q_j$, $P_j$ and $S_j$  are based on robust averages. In detail, we compute trimmed means of $Q_{t_1},\ldots,Q_{t_{N_{j}}}$, $P_{t_1},\ldots,P_{t_{N_{j}}}$ and $S_{t_1},\ldots,S_{t_{N_{j}}}$ to obtain $Q_j$, $P_j$ and $S_j$ respectively. These trimmed means discard the upper 0.5\% and the lower 0.5\% of the corresponding ordered data and compute the average based on the remaining 99\% of the data. 

Thirdly, the estimated squared volatility $\sigma_j^2$ is computed as  realized variance in interval $j$
\begin{align} \label{eq:RealizedVarianceEstimator}
	\hat{\sigma}_j^2= \sum_{k=2}^{N_j} \left(\log (P_{t_k}) - \log (P_{t_{k-1}})\right)^2.
\end{align}
The properties of the estimator $\hat{\sigma}_j^2$ are well understood for a variety of models for the efficient price process $(P_t)_{t\geq 0}$. For example, if the dynamics of the efficient price process follows the stochastic model $d P_t=\sigma P_t dW_t$, with $\sigma>0$, the estimator $\hat{\sigma}_j^2$ converges weakly in probability to $\sigma^2 T$ (the quadratic variation of the increments of $\left(\log(P_t)\right)_{t\geq0}$) as the number of transactions within interval $j$ becomes dense (as $N_j\rightarrow\infty$). The limit of $\hat{\sigma}_j^2$, however, does not coincide with the quadratic variation of the efficient price process, if the observed midquote price is contaminated by market microstructure noise. This noise, for instance, arises from market imperfections such as price discreteness or informational content in price changes, see~\cite{black1986noise}. To check the robustness of our analysis with respect to the presence of market microstructure noise, several results below can likewise be confirmed by replacing the realized variance by the noise-robust estimator of the quadratic variation proposed in~\cite{hautsch2013preaveraging}. It should be noticed that a distortion of the analysis by the bid-ask bounce is already avoided by considering midquote prices rather than transaction prices. The interested reader will find a gentle introduction explaining how noisy price observations erode the realized variance in~\cite{sahalia2009high}.

Last but not least, note that Benzaquen et al.~\cite{benzaquen2016unravelling} in fact define the cost per trade by 
$\widetilde{C}_j =  N_{j}^{-1} \sum_{k=1}^{N_j} Q_{t_k} S_{t_k}$. This slight difference in the definitions becomes obviously negligible, if the bid-ask spread $S_{t_k}$ is constant over the entire interval $j$. The results presented below are robust with respect to the employed version of the cost per trade as we shall see.

\subsection{$N \sim \sigma^2$ versus $N^{3/2} \sim \sigma P V /C$} \label{ssec:Versus}
To check which of the relations $N \sim \sigma^2$ and $N^{3/2} \sim \sigma P V /C$ is superiorly supported by data, we consider for each stock ($i=1,\ldots,d$) a multiplicative model of the form
\begin{align} \label{/sigmaj^2=exp(alpha)..}
	N_{ij} = \exp(\alpha_i) (\hat{\sigma}_{ij}^2)^{\beta_i} \left( \frac{P_{ij} V_{ij}}{C_{ij}}\right)^{\gamma_i} \exp(\varepsilon_{ij}) \quad\text{with}\quad j=1,\dots,n,
\end{align}
where $\varepsilon_{ij}$, $j=1,\ldots,n$, is an error term that satisfies standard regularity conditions and $\alpha_i$, $\beta_i$ and $\gamma_i$ are unknown real valued parameters. A logarithmic transformation of \eqref{/sigmaj^2=exp(alpha)..} yields the linear model
\begin{align}\label{LinReg}
	\log(N_{ij}) = \alpha_i + \beta_i \log\left(\hat{\sigma}_{ij}^2\right) + \gamma_i \log\left(\frac{P_{ij} V_{ij}}{C_{ij}}\right) + \varepsilon_{ij}.
\end{align}
Since dimensional analysis imposes the restriction $\beta_i+\gamma_i=1$ on the parameters $\beta_i$ and $\gamma_i$, the value $\gamma_i=0$  would imply the relation $N\sim \sigma^2$, whereas $\gamma_i = 2/3$ would imply the relation $N^{3/2} \sim \sigma P V /C$ from Theorem \ref{thm:3/2law}. The estimation of the coefficients $\beta_i$ and $\gamma_i$ subject to the restriction $\beta_i+\gamma_i=1$ therefore allows us to infer which of the two discussed relations is backed by stronger empirical evidence.

Before turning to the constrained estimation of the parameters $\beta_i$ and $\gamma_i$, it deserves to be emphasized that the functional relation between the logarithmic dependent variable $\log(N_j)$ and the logarithmic explanatory variable $\log(\hat{\sigma}_{ij} P_{ij} V_{ij}/C_{ij})$ can be reasonably assumed to be linear for all stocks $i=1,\dots,d$. To conclude this, we have visually inspected the bivariate point-clouds of dependent and explanatory variable. Figure~\ref{fig:Scatter} illustrates this relation for the stocks of the American Airline Group, Inc. (AAL) and Apple Inc. (AAPL). The remaining $126$ stocks show similar patterns.
\begin{figure}
\begin{center}
\includegraphics[scale=0.4]{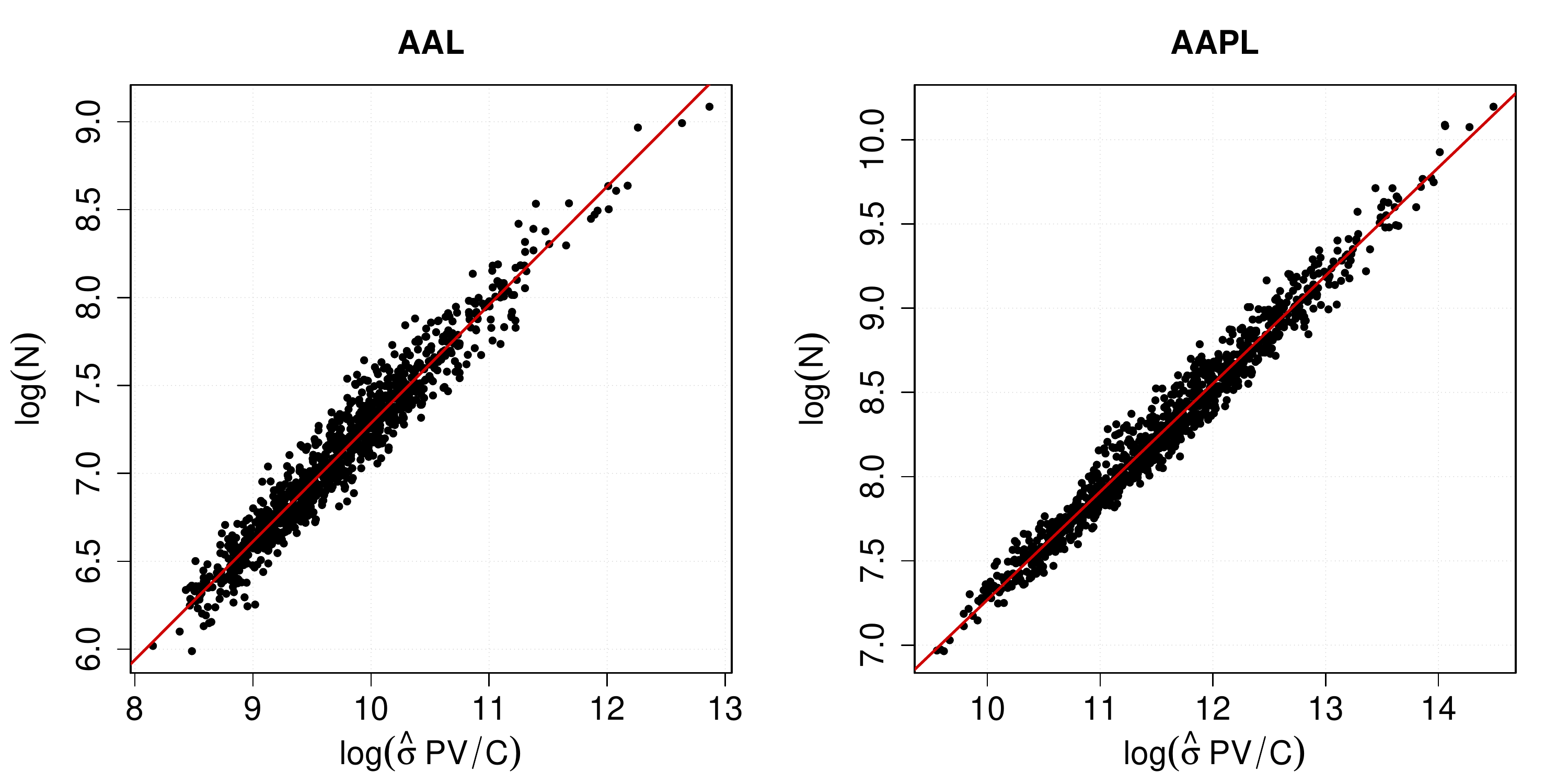}
	\caption{The logarithmic dependent variable $\log(N)$ is plotted versus the logarithmic explanatory variable $\log(\hat\sigma P V /C)$ for the fixed interval length $T = 60$min and the two stocks AAL and AAPL. The lines indicate the estimated linear relations between the considered quantities.}
\label{fig:Scatter}
\end{center}
\end{figure}

For each stock ($i=1,\dots,d$) and all interval lengths $T\in\{30, 60, 120, 180, 360\}$ min, we estimate the parameters $\beta_i$ and $\gamma_i$ in \eqref{LinReg} by ordinary least squares subject to the constraint $\beta_i+\gamma_i=1$. The corresponding estimate of interest is denoted by $\hat{\gamma}_i$. To present the results of these regressions in an informative and compact way, Figure~\ref{fig:den} shows kernel density estimates of $\hat{\gamma}_i$ across $i$ and for fixed $T$.

First, let us come to the main result of this section and concentrate on the solid graphs in Figure~\ref{fig:den} referring to the standard setting based on the realized variance $\hat\sigma_{ij}^2$ defined in \eqref{eq:RealizedVarianceEstimator} and the cost per trade $C_{ij} = Q_{ij} \times S_{ij}$. 
If the parameter $\gamma_i$ of the linear model \eqref{LinReg} is equal to zero, then the underlying variables satisfy the simple relation $N\sim \sigma^2$.
Similarly, if the parameter $\gamma_i$ is equal to $2/3$, then we can conclude that the 3/2-law from Theorem~\ref{thm:3/2law} holds.
 As seen in Figure~\ref{fig:den}, the averages of the estimates $\hat\gamma_i$ (across $i$ for different $T$) are clearly much closer to $2/3$ than to zero for all considered interval lengths $T$. This result supports the claim made in Section~\ref{sec:TradInv} that there is stronger empirical support for the 3/2-law (or equivalently for the relation $N \sim (\sigma P / S)^2$) than for the relation $N \sim \sigma^2$. 

Regarding the robustness of this insight, we have re-conducted the above regression analysis for two slightly different scenarios. One alternative setting considers replacing the realized variance in the linear model \eqref{LinReg} by the market microstructure noise robust estimator of the quadratic variation of~\cite{hautsch2013preaveraging}. The dashed graphs in Figure~\ref{fig:den} are related to density estimates relying on corresponding parameter estimates $\hat\gamma_i$, $i=1,\ldots,d$. The second modification of the initial setting replaces the cost per trade $C_j$ in the linear model \eqref{LinReg} by the variant $\widetilde{C}_j$  of~\cite{benzaquen2016unravelling}. The dotted graphs in Figure~\ref{fig:den} refer to corresponding density estimates. Despite some deviation in the estimates $\hat{\gamma}_i$ for these two alternative settings from the initial one, the solid, dashed and dotted graphs document a rather similar pattern among the estimates of the parameters $\gamma_i$ for all interval lengths $T\in\{30, 60, 120, 180, 360\}$ min. These similarities lead to the conclusion that neither market microstructure noise nor the exact definition of the cost per trade erode the overall relation between the dependent and explanatory variables.  In the remaining part of the manuscript, we take a closer look on the 3/2-law and try to find reasonable explanations for the systematic deviations of the estimates $\hat\gamma_i$ from $2/3$.

\begin{figure}
     \begin{minipage}{1\textwidth}
     \centering
         \includegraphics[width=0.325\textwidth]{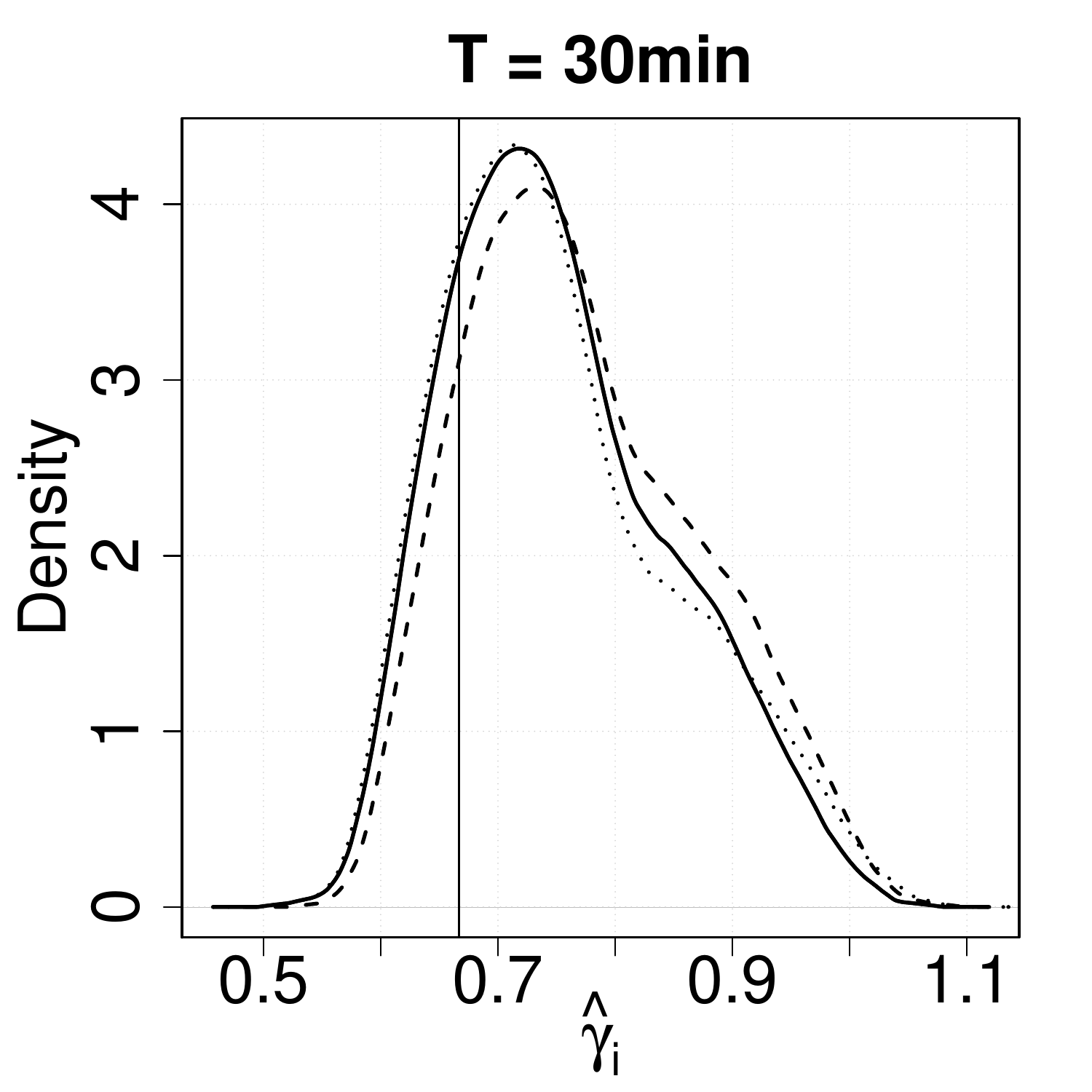}
         \includegraphics[width=0.325\textwidth]{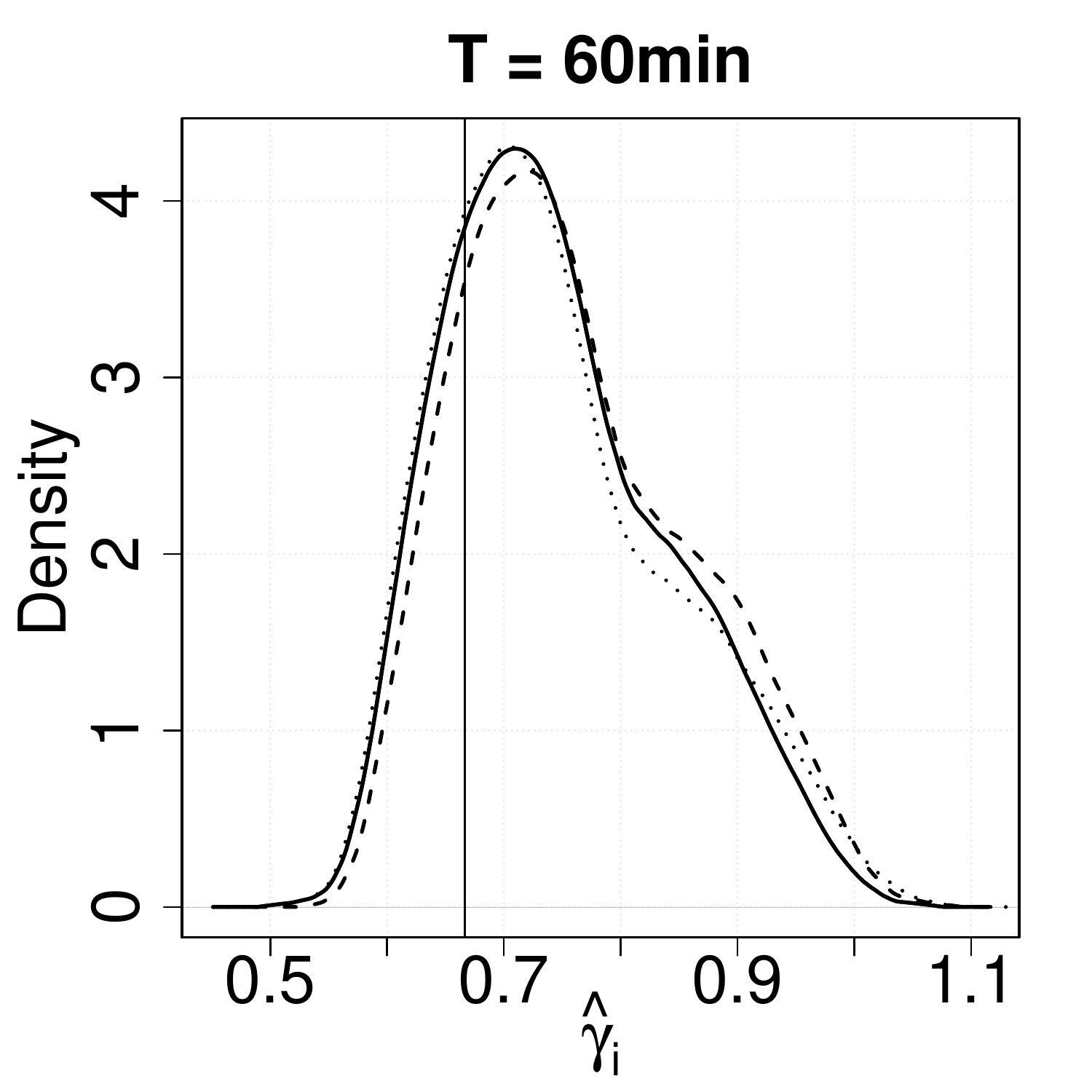}
         \includegraphics[width=0.325\textwidth]{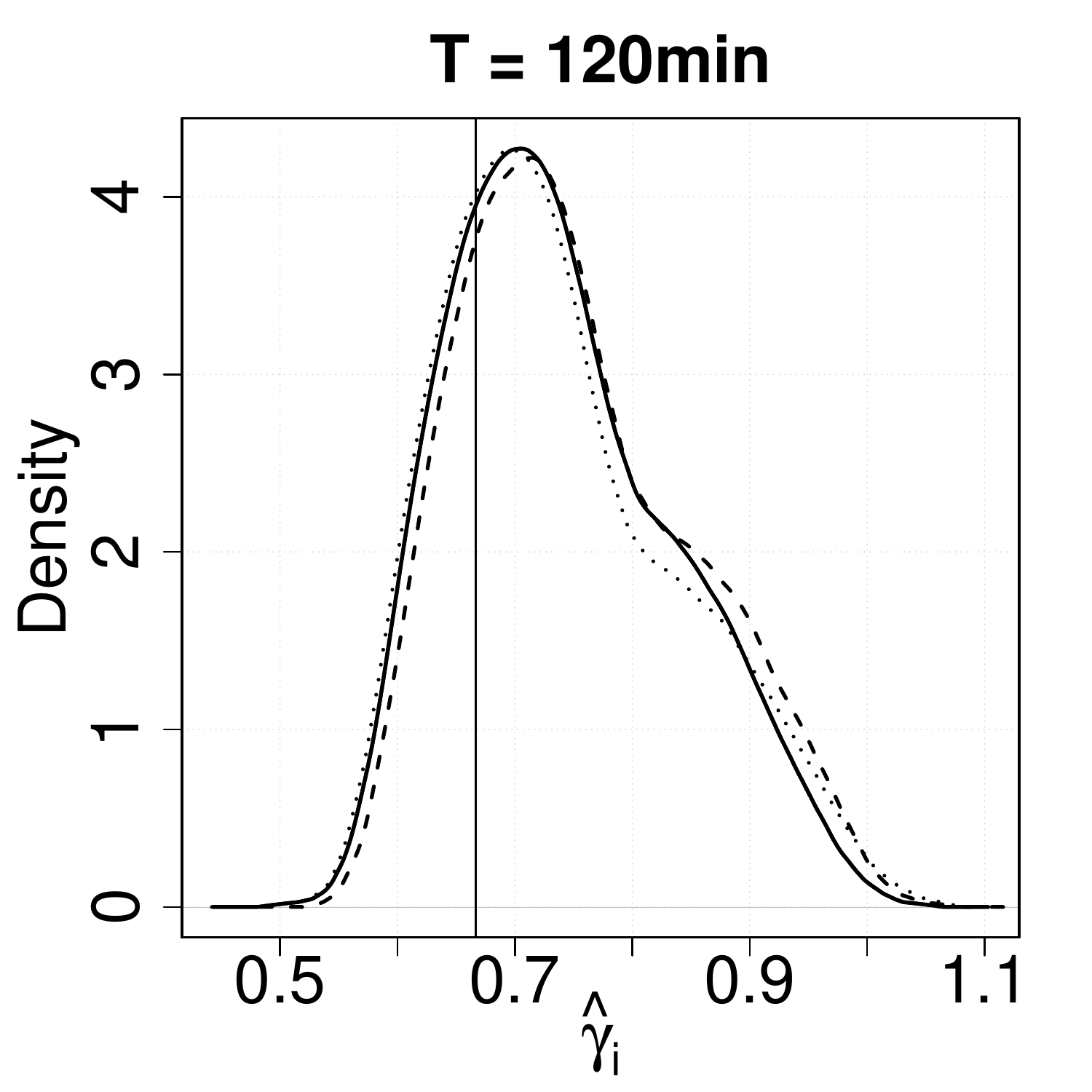}
    \end{minipage}
     \begin{minipage}{1\textwidth}
     \centering
         \includegraphics[width=0.325\textwidth]{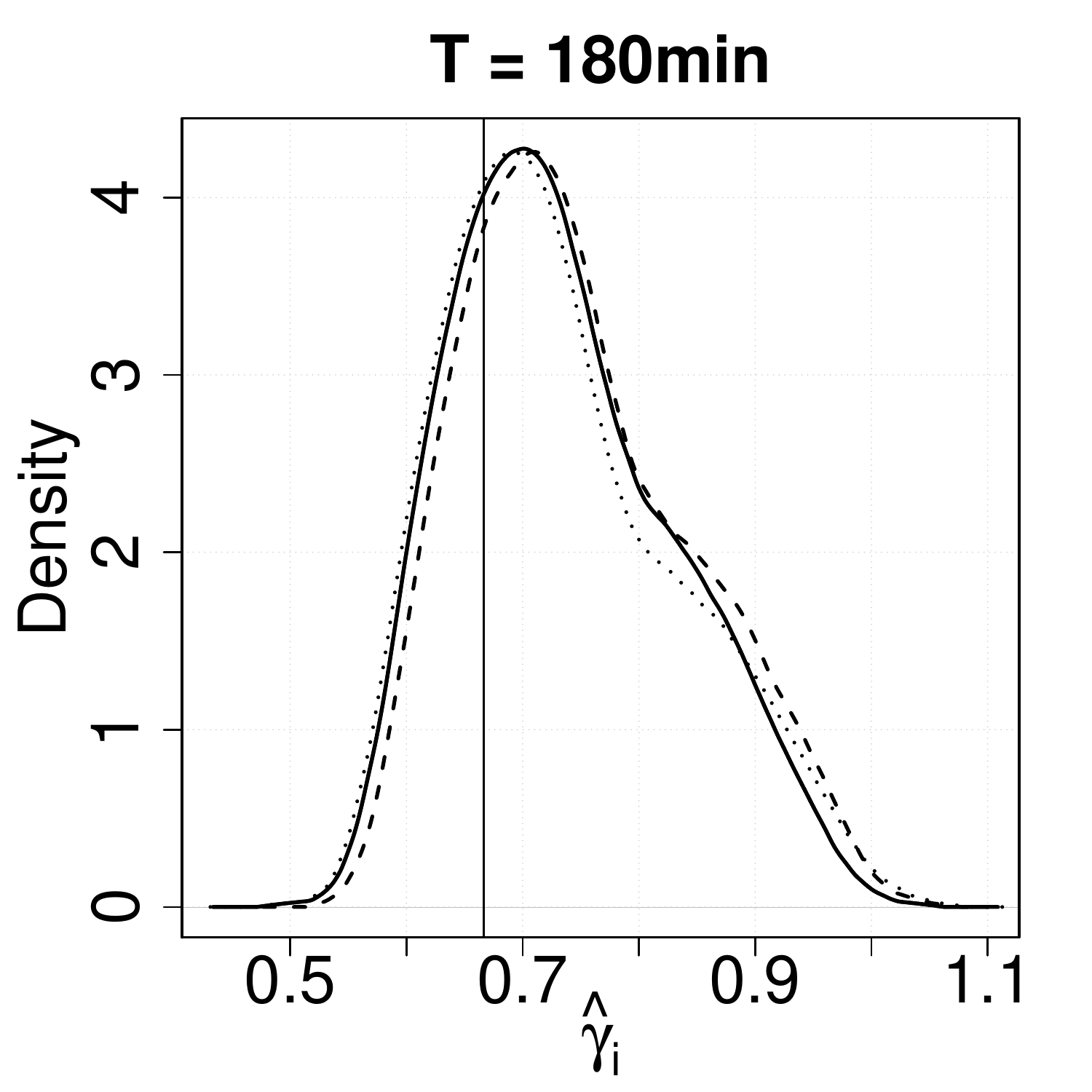}
         \includegraphics[width=0.325\textwidth]{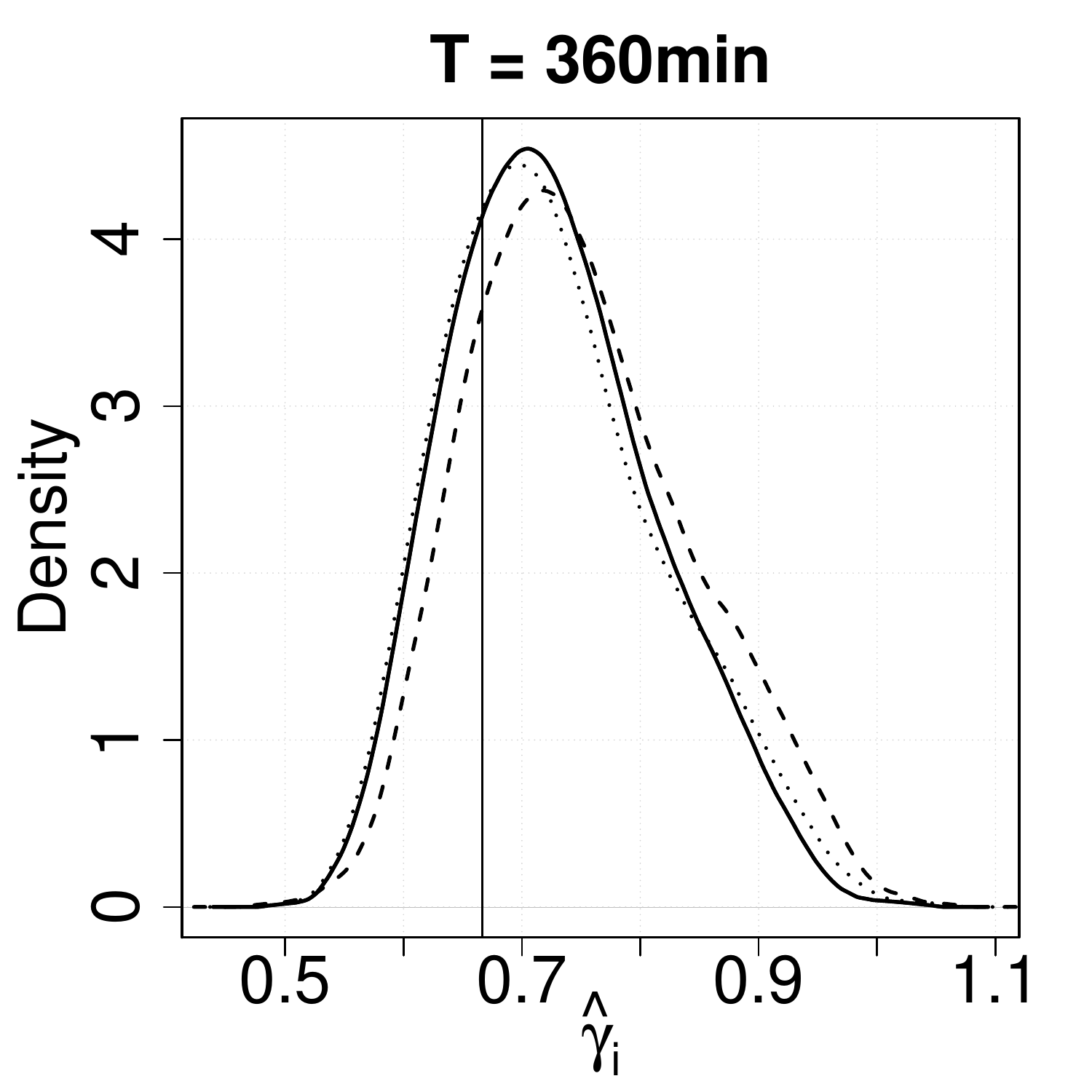}
    \end{minipage}
	\caption{The panels show kernel density estimates across the estimated parameters $\hat{\gamma}_i$ for different interval lengths $T\in\{30, 60, 120, 180, 360\}$ min.}
	\label{fig:den}
\end{figure}

\subsection{On the universality of the 3/2-law}\label{ssec:Universality3/2law}

\begin{figure}
\begin{center}
\includegraphics[scale=0.4]{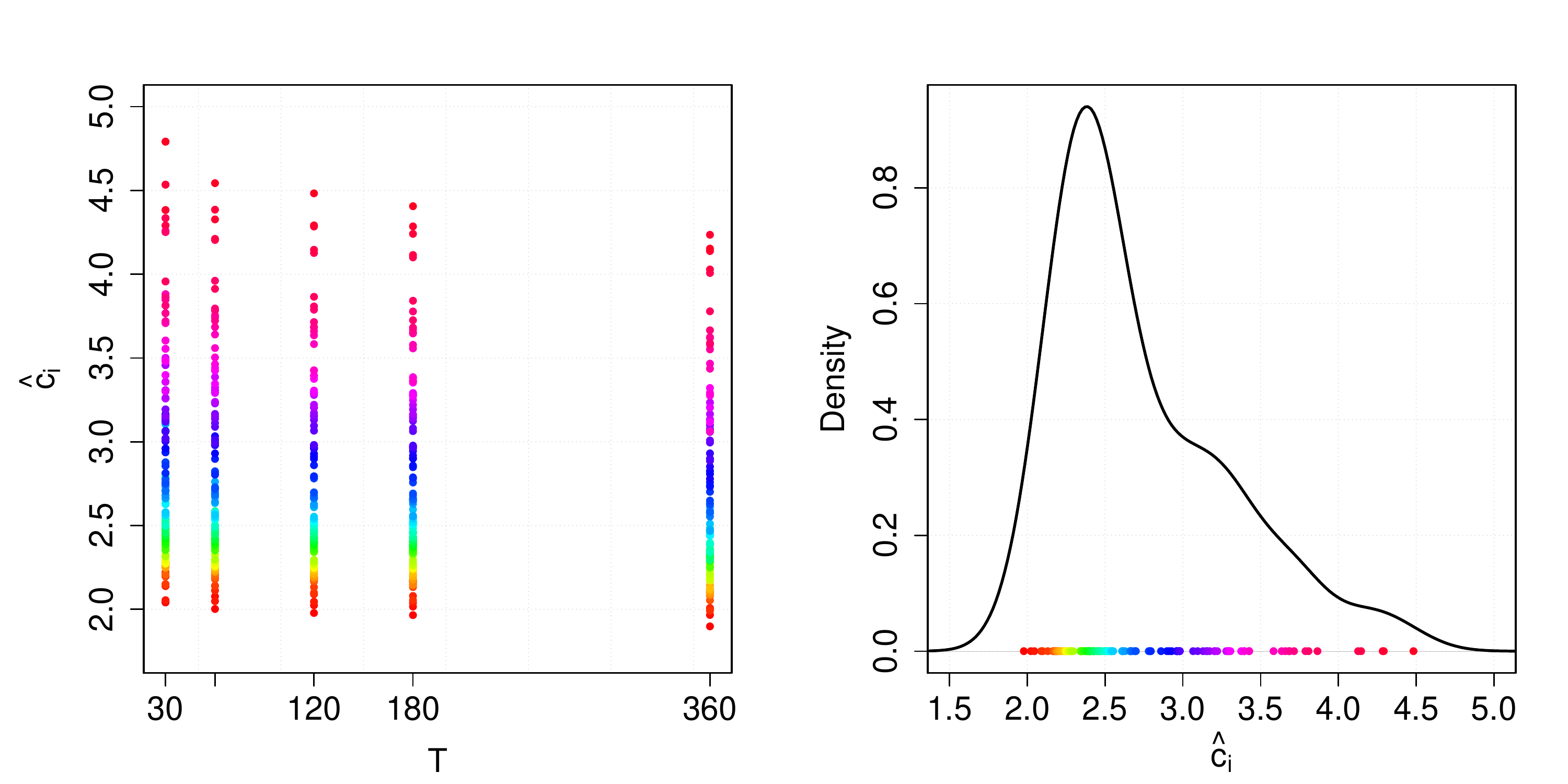}
	\caption{The left panel shows the computed values for $\hat{c}_i$ in dependence of $T\in\{30, 60, 120, 180, 360\}$ min. The right panel shows a kernel density estimate across the estimates $\hat{c}_i$ for fixed $T=120$min.}
\label{fig:rainbow}
\end{center}
\end{figure}

In order to check the validity and universality of the $3/2$-law, $N^{3/2} = c \cdot \sigma P V / C$ (or equivalently of the relation $N = c^2 \cdot (\sigma P / S)^2$), we examine the variation of the constant $c$ across assets and interval lengths. Hence, we do not rely on the estimators $\hat{\gamma}_i$ computed in Section \ref{ssec:Versus}. Instead, we compute for a fixed interval length $T$ the quantity
\begin{align*}
	\hat{c}_i = n^{-1} \sum_{j=1}^n\frac{C_{ij} N_{ij}^{3/2}}{\hat{\sigma}_{ij} P_{ij} V_{ij}}  =n^{-1} \sum_{j=1}^n\frac{N_{ij}^{1/2}}{\hat{\sigma}_{ij}}\frac{S_{ij} }{P_{ij}},\quad\text{for}\quad i=1,\ldots,d,
\end{align*}	
where $n$ is the number of non-overlapping time intervals with equal length $T$. The left panel of Figure~\ref{fig:rainbow} shows the estimates $\hat{c}_i$ for different values of $T$. Note that the  rainbow-color-code refers to the ordered values of $\hat{c}_i$ for $T=120$min. As we recover the same rainbow-pattern also for the other interval lengths $T \in \lbrace 30, 60, 180, 360 \rbrace$ min, we can conclude that there is little variation of the estimates $\hat{c}_i$ for a fixed stock $i$ across different interval lengths $T$. This small variation of $\hat{c}_i$ for fixed $i$ and varying $T\in\{30, 60, 120, 180, 360\}$ min endows the 3/2-law with a certain degree of universality. However, the present cross-sectional dispersion in $\hat{c}_i$ across different assets $i$, i.e., the fact that depending on the considered stock the estimates $\hat{c}_i$ range from two to five, does not allow awarding the 3/2-law with strong universality. Thus, we draw the same conclusion as Benzaquen et al.~\cite{benzaquen2016unravelling} that the 3/2-law holds with weak universality. For completeness, the kernel density estimate in the right panel of Figure~\ref{fig:rainbow} illustrates the distribution of the estimates $\hat{c}_i$, $i=1,\ldots,d$ for $T=120$min.

\section{A closer look on volatility} \label{sec:VolaScaling}
We have seen that the volatility $\sigma$ plays a dominant role in explaining the trading activity $N$. 
The squared volatility $\sigma^2$ of a given stock during a fixed interval $[t,t+T]$ was defined as the variance of the change of the log-price 
\begin{align}
\label{eq:defvola}
\sigma^2 := \mathbb{V}\text{ar} \left(\log(P_{t+T}) - \log(P_t)\right).
\end{align}
When specifying the definition of $\sigma^2$ in this way we had in mind the Black-Scholes model,
\begin{align}
\label{eq:Black-Scholes}
dP_t = P_t \left( \sigma dW_t + \mu dt \right),
\end{align}
where, fixing the normalization $T=1$, formula $\eqref{eq:defvola}$ indeed recovers the constant $\sigma$ in \eqref{eq:Black-Scholes}. Going beyond Black-Scholes, consider a price process of the form
\begin{align}
\label{eq:volProcess}
P_t = P_0 \exp \left( \int_0^t \sigma_u dW_u \right)
\end{align}
where $(\sigma_t)_{t \geq 0}$ is an arbitrary stochastic process (satisfying suitable regularity conditions). In this case, formula \eqref{eq:defvola} should, of course, be interpreted conditionally on the sigma-algebra $\mathcal{F}_t$ and we obtain the ``Wald identity''
\begin{align}
\label{eq:F2a}
\mathbb{V}\text{ar} \left(\log(P_{t+T}) - \log(P_t) \vert  \mathcal{F}_t \right) = \mathbb{E} \left( \int_t^{t+T} \sigma^2_u du \vert \mathcal{F}_t \right).
\end{align}
This implies in particular that, as long as we are in the framework of processes of the form \eqref{eq:volProcess}, the above chosen scaling $$ [\sigma^2] = \mathbb{T}^{-1},$$ is the only reasonable choice.
 
But let us have a closer look at what we are actually doing here. The above reasoning tacitly assumes that we are starting from a \emph{stochastic model} of a price process. The present situation, however, dictates a different point of view: we start from empirical tick data observed  during the interval $[t,t+T]$. Even when we make the heroic assumption that this data is accurately modeled, e.g.~by the Black Scholes model \eqref{eq:Black-Scholes}, the number $\sigma^2$ which we plug into the formula $N=g(\sigma^2,\dots)$ can only be an \emph{estimator} of $\mathcal{\sigma}^2$ obtained from the data at hand. This implies that, strictly speaking, we should write our formulas as $N= g(\hat{\sigma}^2,\dots)$ in dependence of the \emph{estimated} squared volatility $\hat{\sigma}^2$. The gist of the argument is that for the purpose of dimensional analysis the scaling which is relevant is that of the \emph{estimator} of the volatility rather than that of the \emph{true} volatility (whatever this is).
To be concrete, suppose that we are given price data $(P_{t_k})_{k=1,\dots,N}$ for a grid $t\leq t_{1}<\dots<t_N \leq t+T$ in the interval $[t,t+T]$. An obvious choice for the estimator of the squared volatility, which is also used in Section \ref{sec:EmpEvi} above, is 
\begin{align}
\label{eq:quadVariation}
\hat{\sigma}^2 := \sum_{k=2}^N \left( \log(P_{t_k}) - \log(P_{t_{k-1}}) \right)^2.
\end{align}
Clearly, this estimator has the dimension $[\hat{\sigma}^2] = \mathbb{T}^{-1}$ if we suppose that the typical distance $\Delta t_k = t_{k+1}- t_k$ (in absolute terms) does not depend on whether we measure time in seconds or in minutes. Hence, for the estimator $\hat \sigma^2$, the hypothesis $[\hat{\sigma}^2] = \mathbb{T}^{-1}$ underlying the dimensional analysis in Section \ref{sec:TradInv} is satisfied.

However, we can also think of other estimators. Fix $H\in(0,1)$ and define the estimator $\hat{\sigma}^2(H)$ by
\begin{align}
\label{eq:pVariation}
\hat{\sigma}^2(H) := \left( \sum_{k=2}^N \vert \log(P_{t_k}) - \log(P_{t_{k-1}}) \vert^{1/H} \right)^{2H}.
\end{align}
To motivate this estimator, consider the model 
\begin{align}
\label{eq:fBMmodel}
P_t = P_0 \exp(\sigma W_t^H ), \qquad t \geq 0,
\end{align}
where $\sigma>0$ is a fixed number and $(W_t^H)_{t\geq 0}$ is a \emph{fractional} Brownian motion with Hurst parameter $H$, starting at $W_0^H=0$. In this case, the estimator $\hat{\sigma}^2(H)$ in \eqref{eq:pVariation} is a consistent estimator for the parameter $\sigma^2$ in \eqref{eq:fBMmodel}. But the estimator $\hat{\sigma}^2(H)$ now scales differently in time than the quadratic estimator $\hat{\sigma}^2$  (see~\cite{coutin2007introduction,pratelli2011remark}), namely \begin{align} \label{eq:dimensionT2H}
[\hat{\sigma}^2(H)] = \mathbb{T}^{-2H}. 
\end{align}
Models for the price process $(P_t)_{t\geq 0}$ involving fractional Brownian motion as in \eqref{eq:fBMmodel} have been proposed, notably by B. Mandelbrot, already more than 50 years ago~\cite{mandelbrot1963variation,mandelbrot1968fractional} and there may be good reasons not to rule them out a priori.
\vspace{5mm}

Here is another example where a sub-diffusive behavior of the price process $(P_t)_{t^\geq 0}$ occurs, due to a micro-structural effect: the discrete nature of the prices in the real world (compare Benzaquen et al.~\cite{benzaquen2016unravelling}; we thank Jean-Philippe Bouchaud for bringing this phenomenon to our attention). To present the idea in its simplest possible form, suppose that the price process $(\check{P}_t)_{t\geq 0}$ is given by $$ \log(\check{P}_t) = \text{int}(W_t),$$ where $(W_t)_{t \geq 0}$ is a standard Brownian motion and $\text{int}(x)$ denotes the integer closest to the real number $x$, i.e., $\text{int}(x) = \sup \lbrace n \in \mathbb{Z}: n \leq x + 0.5 \rbrace$.
Fix again an interval $[t,t+T]$ and consider the quantity 
\begin{align*}
\check{\sigma}^2 = (\check{\sigma}^2)_t^{t+T} = \mathbb{V}\text{ar} \left( \log(\check{P}_{t+T}) - \log(\check{P}_t)\right).
\end{align*}
For small $T>0$, we show in Appendix \ref{sec:RoundedBM} that
$$  (\check{\sigma}^2)_t^{t+T}  \approx \text{c }\sqrt{T},$$ for some constant $c>0$.
Hence, if the interval length $T$ is sufficiently small, we recover that $[\check{\sigma}^2 ] = \mathbb{T}^{-1/2}$, rather than the usual scaling in the dimension time, i.e., $\mathbb{T}^{-1}$.
 
This observation indicates, that if the interval length $T$ is small compared to the width of the price grid, i.e., the tick value, we observe a sub-diffusive behavior of the price process even if the ``efficient'', unobserved price process is assumed to be a diffusion. We refer to Robert and Rosenbaum~\cite{robert2010new} for a detailed discussion of how to account for the discrete nature of prices. For now, this rough argument should only serve as motivation that there might be plenty of reasons why the scaling $[\sigma^2] = \mathbb{T}^{-1}$ is, in practical situations, not as clearly granted as it might seem at first glance.

For all these reasons we drop in this section the convenient dimensional assumption $[\sigma^2] = \mathbb{T}^{-1}$ and replace it by the subsequent more general assumption.
\begin{Hassumption}
There is $H \in (0,1)$ such that the squared volatility estimator $\hat{\sigma}^2(H)$ has dimension 
$$[\hat{\sigma}^2(H)] = \mathbb{T}^{-2H}.$$
\end{Hassumption}
\begin{proposition}[$(1+H)$-law]
\label{prop:1+Hlaw}
Suppose that the ``Leverage Neutrality Assumption''  as well as the ``$H$-Assumption'' hold true and that the number of trades $N$ depends \emph{only} on the four quantities $\hat{\sigma}^2(H), P, V$ and $C$, i.e.,
\begin{align*}
	N &= g(\hat{\sigma}^2(H),P,V,C),
\end{align*}
 where the function $g:\mathbb{R}_+^4\rightarrow\mathbb{R}_+$ is \emph{dimensionally invariant} and \emph{leverage neutral}. Then, there is a constant $c>0$ such that the number of trades $N$ obeys the relation
\begin{align}\label{eq:1+Hlaw}
	N^{1+H} =c\, \cdot\, \frac{\hat{\sigma}(H) PV}{C}.
\end{align}
\end{proposition}

The proof is analogous to the proof of Theorem \ref{thm:3/2law} and is given in Appendix \ref{sec:Proofs}. 

The hypothesis of the above proposition assumes that $H \in (0,1)$ is known a priori. As $H$ is typically unknown in practical applications, we can therefore ask the following question: For which $H$ does relation \eqref{eq:1+Hlaw} fit the empirical data best? We address this question in the following subsection.

\subsection{Empirical evidence under the $H$-Assumption}

\begin{figure}
\begin{center}
        \includegraphics[scale=0.45]{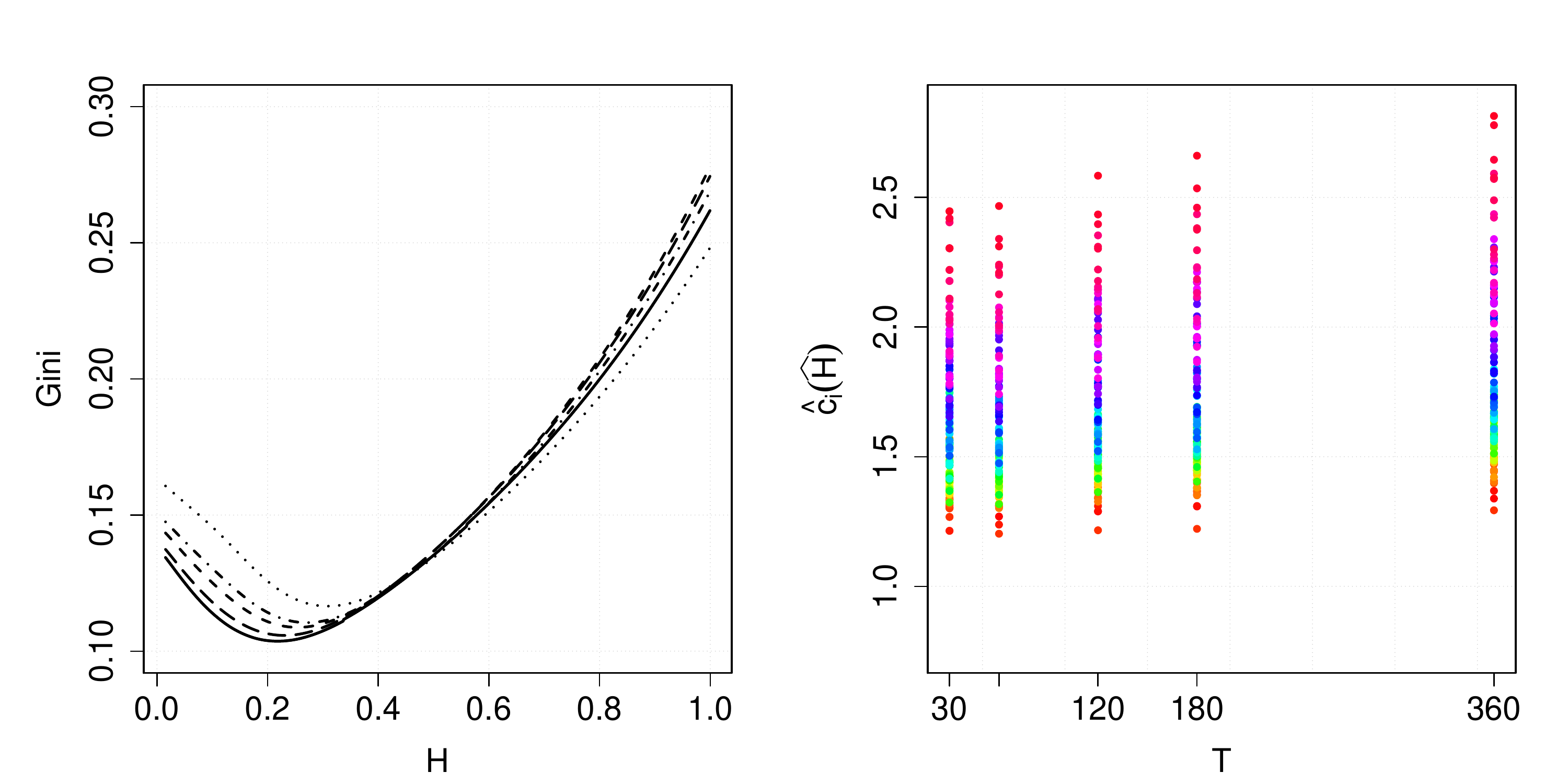}
	\caption{The left panel illustrates the Gini-coefficient in dependence of $H$ for $T=30$min (solid), $T=60$min (long-dashed), $T=120$min (dashed), $T=180$min (dashed-dotted) and $T=360$min (dotted). The right panel shows the computed values for $\hat{c}_i(\widehat{H})$ such that $\widehat{H}$ minimizes the Gini-coefficient for fixed $T\in\{30, 60, 120, 180, 360\}$ min.}
\label{fig:rainbowH}
\end{center}
\end{figure}

According to arguments from dimensional analysis, the constant $c$ and the parameter $H$ from Equation~\eqref{eq:1+Hlaw} should at best be identical for all stocks and all interval lengths $T$. The empirical results above, however, have revealed cross-sectional dispersion which might be related to the restrictive assumption $[\hat{\sigma}^2]=\mathbb{T}^{-1}$. This restriction motivates the empirical exercise of this section: Can we determine an $H\in(0,1)$ in \eqref{eq:1+Hlaw} that minimizes the cross-sectional dispersion across the estimates of $c$?

Following Proposition \ref{prop:1+Hlaw}, we therefore compute the estimates $\hat{c}_i(H)$ for different $H$ as
\begin{align*}
	\hat{c}_i(H) = n^{-1} \sum_{j=1}^n\frac{N_{ij}^{1+H}C_{ij} }{\hat{\sigma}_{ij}(H) P_{ij} V_{ij}} =  n^{-1} \sum_{j=1}^n\frac{N_{ij}^{H}}{\hat{\sigma}_{ij}(H)}\frac{S_{ij} }{P_{ij}},\quad\text{for}\quad i=1,\ldots,d,
\end{align*}	
where $\hat{\sigma}_{ij}^2(H)$ is defined in \eqref{eq:pVariation}, $H \in (0,1)$.
Both variables $N_{ij}^{H}$ and $\hat{\sigma}_{ij}(H)$ increase as $H$ increases, so that it is not obvious how $\hat c_i(H)$ behaves when $H$ increases.
We find empirically that overall the constant $\hat{c}_i(H)$ typically increases in $H$. Addressing the above question therefore requires a scale invariant measure for the variation in $\hat{c}_i(H)$ such as the Gini-coefficient which is given by
\begin{align*}
	\mathcal{G}(x_1,\ldots,x_n) &=\frac{2 \sum_{i=1}^n i x_{[i]}}{(n-1)\sum_{i=1}^nx_{[i]}}-\frac{n+1}{n-1},
\end{align*}
for the ordered data $x_{[1]}<x_{[2]}<\ldots<x_{[n]}$. Note that the Gini-coefficient $\mathcal{G}(x_1,\ldots,x_n)\in[0,1]$ is interpreted as a measure for inequality. If all values $x_1,\ldots,x_n$ are equal, $\mathcal{G}$ equals zero. In case of strong heterogeneity in $x_1,\ldots,x_n$ the Gini-coefficient approaches one.\footnote{The coefficient of variation defined as the ratio of the standard deviation to the sample average could be employed as an alternative to the Gini-coefficient. The presented results are widely robust with respect to the chosen measure of standardized dispersion.}

Now, we minimize the Gini-coefficient of $\left(\hat{c}_i(H) \right)_{i=1,\dots,d}$ with respect to $H$ in order to find
\begin{align*}
	\widehat{H} = \underset{H\in(0,1)}{\text{arg min}}\;\mathcal{G}(\hat{c}_1(H),\ldots,\hat{c}_n(H)).
\end{align*}
The left panel of Figure~\ref{fig:rainbowH} plots the Gini-coefficient in dependence of $H$ for different interval length $T$. We roughly find that $\widehat{H}=0.22$ for $T=30$min, $\widehat{H}=0.23$ for $T=60$min, $\widehat{H}=0.25$ for $T=120$min, $\widehat{H}=0.27$ for $T=180$min and $\widehat{H}=0.31$ for $T=360$min.
The rainbow-color-code of Figure~\ref{fig:rainbow} has been transferred to the right panel of Figure~\ref{fig:rainbowH}. In contrast to Figure~\ref{fig:rainbow} yet, we present the quantities $\hat{c}_i(\widehat{H})$ in dependence of the optimal $\widehat{H}$ for the given interval length $T$. In case $T=120$min for instance, the estimates $\hat{c}_i(H=0.25)$ range from 1.2 to 2.6 for different assets $i$. On an absolute scale, the variation seems to be smaller compared to Figure~\ref{fig:rainbow}, where the estimates $\hat{c}_i(H=0.5)$ lie between 2 and 4.5 for the same interval length $T=120$min. In relative terms though, the difference between the variation in $\hat{c}_i(H=0.25)$ and $\hat{c}_i(H=0.5)$ is not so significant, as $\mathcal{G}\left(\hat{c}_1(H=0.25),\dots,\hat{c}_n(H=0.25) \right) = 0.11$ compared to $\mathcal{G}\left(\hat{c}_1(H=0.5),\dots,\hat{c}_n(H=0.5) \right) = 0.14$ for $T=120$min.

For now, we can only speculate on reasons why the optimal $\widehat{H}$ is strikingly smaller than $1/2$ for all interval lengths $T$. The quantity $\hat{c}_i(H)$ relies on tick-by-tick data, so that an
obvious explanation for these unexpected optimal values of $H$ are market microstructure effects. To be more concrete, Benzaquen et al.~\cite{benzaquen2016unravelling} observe similar to our results a sub-diffusive behavior for so called large tick future contracts. Large tick assets are defined such that their bid-ask spread is almost always equal to one tick, see e.g.~\cite{eisler2012price}. Most of the stocks in our sample can be categorized as large tick stocks based on this definition.

When referring to market microstructure effects, however, it deserves to be stressed that the value $H=1/2$ is implied by numerous models for the efficient price process $(P_t)_{t\geq0}$, which are backed by empirical evidence and take market microstructure effects into account. Hence, the scaling of the squared volatility through time implied by $H=1/2$ seems suitable in many applications. We also note that the Gini-coefficient $\mathcal{G}$ in Figure~\ref{fig:rainbowH} does not vary drastically when $H$ ranges between the optimal $\widehat H \approx 0.25$ and the traditional $H=1/2$, namely roughly between $\mathcal{G}= 0.12$ and $\mathcal{G}= 0.15$.
Hence, the value of $H$ does not seem to play a very significant role in explaining the heterogeneity of the value of $\hat c_{ij}(H)$.
Nevertheless, a better understanding of the behavior of $\widehat H$ seems to us a challenging topic for future research.

\section{Conclusion}

Finding laws relating the trading activity (defined here as the number of trades $N$ within a given time interval) to other relevant market quantities has been the subject of numerous investigations.
The earliest contribution dating as far back as the beginning of the 1970s. Two decades later, Jones et al.~\cite{jones1994transactions} suggested the relation $N\sim \sigma^2$ based on an extensive empirical study.
Other landmark contributions include the relation $N \sim (\sigma P/S)^2$ of Madhavan et al.~\cite{madhavan1997why} resp. Wyart et al.~\cite{wyart2008relation} and the so called $3/2$-law $N^{3/2}\sim \sigma P V/C$ of Benzaquen et al.~\cite{benzaquen2016unravelling}, which were obtained using market microstructure arguments and supported by empirical evidence.
In the first part of the paper we show that all these scaling laws can be derived using arguments relying on dimensional analysis.
The relation $N\sim \sigma^2$ follows from the assumption that $N$ is fully explained by the squared volatility $\sigma^2$, the asset price $P$ and the traded volume $V$, and the assumption that the relation between these quantities is invariant under changes of the dimensions shares $\mathbb{S}$, time $\mathbb{T}$ and money $\mathbb{U}$.
The somewhat refined relation $N^{3/2}\sim \sigma PV/C$ is obtained when assuming that $N$ depends only on $\sigma^2, P, V$ and the cost of trading $C$, and assuming in addition, that an invariance principle known as ``Leverage Neutrality'' holds true. 
This ``Leverage Neutrality Assumption'' can be seen as a no-arbitrage condition enabling us to obtain a unique functional relation from the assumption $N= g(\sigma^2, P, V, C)$.
Substituting the quantity $C$ by the bid-ask spread $S$ in the latter assumption, we derive the relation $N \sim (\sigma P/S)^2$, which is shown to be equivalent to the $3/2$-law.
Alternatively, we can consider the volatility of the \emph{relative} price change instead of the \emph{absolute} price change, i.e., assume $N = g(\sigma^2P^2, V, C)$ resp.~$N = g(\sigma^2P^2, V, S)$.
This assumption simplifies the analysis in that a unique solution for $g(\cdot,\cdot,\cdot)$ can be obtained without recourse to the ``Leverage Neutrality Assumption''.
Since our \emph{theoretical} analysis relies on a set of well-defined, but not necessarily realistic assumptions, the validity of any of the aforementioned scaling laws needs to be confirmed  through an empirical analysis.

Based on data from the NASDAQ stock exchange, we provide empirical evidence that the $3/2$-law $N^{3/2}=c\cdot \sigma PV/C$ (or equivalently $N = c^2\cdot (\sigma P/S)^2$) fits the data clearly better than $N\sim \sigma^2$.
In fact, the $3/2$-law holds for a fixed asset and a fixed interval length.
However, the estimated value of the constant $c$  strongly depends on the considered asset.
In the language of Benzaquen et al.~\cite{benzaquen2016unravelling}, this means that the $3/2$-law holds with weak universality.

Finally, we note that both our theoretical and empirical analysis relied on the assumption that the scaling of $\sigma^2$ is inversely proportional to time $\mathbb{T}$.
This hypothesis is clearly debatable as it tacitly assumes diffusive price behaviors, and ignores e.g.~the discrete nature of prices.
A closer look at the scaling of $\sigma^2$ suggests the scaling $[\sigma^2] = \mathbb{T}^{-2H}$ for some $H \in (0,1)$ that can be seen e.g.~as the Hurst parameter of a fractional Brownian motion.
Repeating our dimensional arguments, the latter scaling of $\sigma^2$ yields the relation $N^{1+H}\sim \sigma^2PV/C$.
An essential drawback of this more general situation is that the parameter $H$ is unknown.
We formulate an optimality criterion for the choice of $H$.
It should yield the most homogeneous estimates for the proportionality coefficients $\hat c_i(H)$.
A preliminary analysis implies that, on average, the optimal $\widehat H$ is of the order $0.25$, i.e., quite different from the assumption $H= 0.5$.
Although the overall effect of this passage from $H=0.5$ to $\widehat H \approx 0.25$ turns out to have only mild effects on the issue of universality of the corresponding laws, we believe that this phenomenon merits further investigation.

\appendix
\section{Dimensional analysis and the Pi-Theorem}
\label{sec:DAandPiTheorem}
In order to formally prove the results of Sections \ref{sec:TradInv} and \ref{sec:VolaScaling}, which in done in Appendix \ref{sec:Proofs}, we need the Pi-Theorem from dimensional analysis. For completeness, we therefore provide the following reminder of this important theorem from dimensional analysis, which can also be found in~\cite{pohl2017amazing}.  Additionally, the interested reader is referred to 
Chapter 1 of the book by Bluman and Kumei~\cite{bluman2013symmetries}
as well as to Pobedrya and Georgievskii~\cite{pobedrya2006proof} for a historical perspective and to~\cite{curtis1982dimensional} for a purely mathematical treatment of dimensional analysis. We formalize the assumptions behind dimensional analysis in proper generality. However, for the purpose of the present paper we shall only need the degree of generality covered by Corollaries~\ref{thm:pi,k=0} and \ref{thm:pi,k=1} below.
\vspace{5mm}

\begin{assumption}[Dimensional analysis]\label{ass:DimAnal} ~ \begin{enumerate}[(i)]
\item Let the quantity of interest $U \in \mathbb{R}_+$ depend on $n$ quantities $W_1,\dots,W_n \in \mathbb{R}_+$, i.e.,
\begin{align}
	U = h(W_1,W_2,\dots,W_n), \label{ProblemGeneral}
\end{align}
for some function $h:\mathbb{R}_+^n \to \mathbb{R}_+$.
\item The quantities $U,W_1,\dots,W_n$ are measured in terms of $m$ fundamental dimensions labelled $L_1,\dots,L_m$, where $m \leq n$. For any positive quantity $X$, its dimension $[X]$ satisfies $[X] = L^{x_1}_1\cdots L^{x_m}_m$ for some $x_1,\dots,x_m \in \mathbb{R}$. If $[X]=1$, the quantity $X$ is called \emph{dimensionless}.

The dimensions of the quantities $U,W_1,W_2,\dots,W_n$ are known and given in the form of vectors $a$ and $b^{(i)}\in\mathbb{R}^m$, $i = 1, \dots, n$, satisfying $[U] = L_1^{a_1}\cdots L_m^{a_m}$ and $[W_i] = L_1^{b_{1i}}\cdots L_m^{b_{mi}}$, $i = 1,\dots,n$.
Denote by $B = (b^{(1)},b^{(2)}, \dots,b^{(n)})$ the $m \times n$ matrix with column vectors $b^{(i)}= (b_{1i},\dots, b_{mi})^\top$, $i=1,\ldots,n$.
\item For the given set of fundamental dimensions  $L_1,\dots,L_m$, a \emph{system of units} is chosen in order to measure the value of a quantity.  A change from one system of units to another amounts to rescaling all considered quantities. In particular, dimensionless quantities remain unchanged and formula \eqref{ProblemGeneral} is invariant under arbitrary scaling of the fundamental dimensions.
\end{enumerate}
\end{assumption}
We can now state the main result from dimensional analysis (see~\cite{bluman2013symmetries}).
\begin{theorem}[Pi-Theorem]
\label{thm:pi}
Under Assumption \ref{ass:DimAnal}, let $x^{(i)} := (x_{1i},\dots,x_{ni})^\top$, $i=1,\dots,k:= n-\operatorname{rank}(B)$ be a basis of the solutions to the homogeneous system $Bx=0$ and $y:= (y_1,\dots,y_n)^\top$ a solution to the inhomogeneous system $By = a$ respectively. Then, there is a function $f:\mathbb{R}_+^k \to \mathbb{R}_+$ such that 
\begin{align*}
	U  \cdot W_1^{-y_1} \cdots W_n^{-y_n} =  f(\pi_1,\dots,\pi_k),
\end{align*}
where $\pi_i := W_1^{x_{1i}} \cdots W_n^{x_{ni}}$ are dimensionless quantities, for $i=1,\dots,k$.
\end{theorem}

We shall only need the special cases $k=0$ and $k=1$, which are spelled out in the two subsequent corollaries.

\begin{corollary}
\label{thm:pi,k=0}
Under Assumption~\ref{ass:DimAnal}, suppose that $\operatorname{rank}(B)=n$ and let $y:=(y_1,\dots, y_n)^\top$ be the unique solution to the linear system $By = a$. Then there is a constant $\operatorname{const}>0$ such that 
\begin{align*}
U = \operatorname{const} \cdot\,W_1^{y_1} \cdots W_n^{y_n}.
\end{align*}
\end{corollary}

\begin{corollary}
\label{thm:pi,k=1}
Under Assumption~\ref{ass:DimAnal}, suppose that $\operatorname{rank}(B) = n - 1$ and let $x:= (x_1,\dots, x_n)^\top$ and $y:=(y_1,\dots, y_n)^\top$ be non-trivial solutions to the homogeneous and inhomogeneous systems $Bx = 0$ and $By = a$ respectively. Then there is a function $f: \mathbb{R}_+ \rightarrow \mathbb{R}_+$ such that
\begin{align*}
U = f(W_1^{x_{1}} \cdots W_n^{x_n}) W_1^{y_1} \cdots W_n^{y_n}.
\end{align*}
\end{corollary}

\section{Proofs of Sections \ref{sec:TradInv} and \ref{sec:VolaScaling}}
\label{sec:Proofs}

In this section, we provide formal arguments for the results presented in Sections \ref{sec:TradInv} and \ref{sec:VolaScaling}.
The proofs are based on Corollaries \ref{thm:pi,k=0} and \ref{thm:pi,k=1} above.

\begin{proof}[Proof of Proposition \ref{pro:naive}]
	Combining relation \eqref{eq:naive} and the dimensions of the quantities $\sigma^2,P, V$ and $N$, we obtain that the matrix $B$ as well as the vector $a$ are given by
	\begin{align*}
	B = \left( \begin{array}{ccc}
	\,\,\,\,0 & -1 &\,\,\,\, 1  \\
	 \,\,\,\, 0& \,\,\,\,1 &\,\,\,\, 0   \\
	 -1 &\,\,\,\,0 &-1  \\
	\end{array}
	\right) \quad\text{and}\quad a = \left( \begin{array}{c}
	\,\,\,\,0 \\ \,\,\,\,0 \\ -1
	\end{array} \right).
	\end{align*}
	Table \ref{tab:Dimensions} illustrates how $B$ and $a$ relate to the considered quantities and their dimensions.
	As the matrix $B$ has full rank, i.e., rank$(B) = 3$, applying Corollary \ref{thm:pi,k=0} yields
	\begin{align*}
		N= c\,\cdot\,\sigma^{2y_1} P^{y_2}V^{y_3},
	\end{align*}
	for some constant $c>0$, where $y=(y_1, y_2, y_3)^\top$ is the unique solution of the linear system $By=a$ which is given by $y=\left(1,0,0\right)^\top$.
\end{proof}

\begin{proof}[Proof of Relation \eqref{eq:UnknownFunctionf}]
	Combining relation \eqref{eq:N=g(sigma^2,P,V,C)1} and the dimensions of the quantities $\sigma^2,P, V $ and $C$ as well as $N$, the matrix $B$ as well as the vector $a$ become
	\begin{align*}
		B = \left( \begin{array}{cccc}
		\,\,\,\,0 & -1 &\,\,\,\,1 & 0   \\
 		\,\,\,\,0 & \,\,\,\,1 &\,\,\,\, 0&1   \\
 		-1 &\,\,\,\,0 &-1 &0 \\
		\end{array}
		\right) \quad\text{and}\quad a = \left( \begin{array}{c}
		\,\,\,\,0 \\ \,\,\,\,0 \\ -1
		\end{array} \right).
	\end{align*}
	The vector $x= (-1,1,1,-1)^\top$ is a solution of the homogeneous system $Bx=0$, and the vector $y = (1,0,0,0)^\top$ is a solution of the inhomogeneous system $By=a$.
	Thus, relation \eqref{eq:UnknownFunctionf} follows from Corollary \ref{thm:pi,k=1}.
\end{proof}
\begin{proof}[Proof of Theorem \ref{thm:3/2law}]
	Combining the dimensions of the quantities considered in relation \eqref{eq:N=g(sigma^2,P,V,C)THM} and the ``Leverage Neutrality Assumption'', we obtain that the matrix $B$ as well as the vector $a$ are given by
	\begin{align*}
		B = \left( \begin{array}{cccc}
		\,\,\,\,0 & -1 &\,\,\,\,1 & 0   \\
 		\,\,\,\,0 & \,\,\,\,1 &\,\,\,\, 0&1   \\
 		-1 &\,\,\,\,0 &-1 &0 \\
		\,\,\,\,\,2 & -1 & \,\,\,\,\,0 & 0  \\
		\end{array}
		\right) \quad\text{and}\quad a = \left( \begin{array}{c}
		\,\,\,\,\,0 \\ \,\,\,\,\,0 \\ -1 \\ \,\,\,\,\, 0
		\end{array} \right).
	\end{align*}
	As the matrix $B$ has full rank, i.e., rank$(B) = 4$, applying Corollary \ref{thm:pi,k=0} yields
	\begin{align*}
		N_t= c\,\cdot\, \sigma_t^{2y_1}P_t^{y_2}V_t^{y_3} C_t^{y_4},
\end{align*}
for some constant $c>0$, where $y=(y_1, y_2, y_3, y_4)^\top$ is the unique solution of the linear system $By=a$ which is given by $y=\left(1/3,2/3, 2/3,  -2/3 \right)^\top$.
\end{proof}

\begin{proof}[Proof of Corollary \ref{coro:Corollary3/2Bachelier}]
	Considering the dimensions of the quantities $\sigma_B,V,C$, we obtain that the matrix $B$ as well as the vector $a$ are given by
	\begin{align*}
	B = \left( \begin{array}{ccc}
	        -2 & \,\,\,\,1 & 0    \\
	 \,\,\,\,2 & \,\,\,\,0 & 1   \\
	        -1 &        -1 & 0  \\
	\end{array}
	\right) \quad\text{and}\quad a = \left( \begin{array}{c}
	\,\,\,\,0 \\ \,\,\,\,0 \\ -1
	\end{array} \right).
	\end{align*}
	As the matrix $B$ has full rank, i.e., rank$(B) = 3$, applying Corollary \ref{thm:pi,k=0} yields
	\begin{align*}
		N= c\,\cdot\, V^{y_1}\sigma_B^{y_2}C^{2y_3},
	\end{align*}
	for some constant $c>0$, where $y=(y_1, y_2, y_3)^\top$ is the unique solution of the linear system $By=a$ which is given by $y=\left(1/3,2/3,-2/3\right)^\top$.
	This shows \eqref{eq:3/2BachelierVersion}.
\end{proof}

\begin{proof}[Proof of Corollary \ref{cor:sigma2OverS}]
	As explained before the statement of Corollary \ref{cor:sigma2OverS}, the conditions \eqref{eq:N=g(sigma^2,P,V,C)THM} and \eqref{eq:N=g(sigma_B^2,V,S)} are equivalent.
	Thus, it holds
	\begin{equation*}
		N^{3/2} = c\cdot \frac{\sigma_BV}{C}.
	\end{equation*}
	Since $C=SV/N$, the corollary follows.
\end{proof}
\begin{proof}[Proof of Proposition \ref{prop:1+Hlaw}]
	The proof is the same as that of Theorem \ref{thm:3/2law} except that in the present case the matrices $B$ and $a$ are given by
	\begin{align*}
		B = \left( \begin{array}{cccc}
		\,\,\,\, 0 &        -1 & \,\,\,\,1 & 0   \\
 		\,\,\,\, 0 & \,\,\,\,1 & \,\,\,\,0 & 1  \\
 		\,\,\,\,-2H& \,\,\,\,0 &        -1 & 0  \\
		\,\,\,\, 2 &        -1 & \,\,\,\,0 & 0  \\
		\end{array}
		\right) \quad\text{and}\quad a = \left( \begin{array}{c}
		\,\,\,\,\,0 \\ \,\,\,\,\,0 \\ -1 \\ \,\,\,\,\, 0
		\end{array} \right).
	\end{align*}
	The unique solution $y$ of the linear system $By = a$ is $y = 1/(1+H)\cdot(1/2,1,1,-1)^\top$.
	Applying Corollary \eqref{thm:pi,k=0} gives the desired result.
\end{proof}

\section{Integer part of Brownian motion}
\label{sec:RoundedBM}
With the notation from Section \ref{sec:VolaScaling}, we want to show that as $T\searrow 0$ $$ \mathbb{V}\text{ar} \left( \log(\check{P}_{t+T}) - \log(\check{P}_t) \right)  \approx c \sqrt{T},$$
for some constant $c>0$. Recall that $\left(\log(\check{P}_t)\right)_{t\geq 0}$ is given by  $$ \log(\check{P}_t) = \text{int}(W_t),$$ where $(W_t)_{t \geq 0}$ is a standard Brownian motion and $\text{int}(x)$ denotes the integer closest to the real number $x$, i.e., $\text{int}(x) = \sup \lbrace n \in \mathbb{Z}: n \leq x + 0.5 \rbrace$.
\vspace{5mm}

To present the idea in its simplest possible form, note that for fixed $t>0$, say $t=1$ and $T$ small, it is straightforward to verify that 
\begin{align*}
\left(\log(\check{P}_{t+T}) - \log(\check{P}_t) \right)^2 = \left(\text{int}(W_{t+T}) -\text{int}(W_t) \right)^2 = \begin{cases} 0 &\hspace*{-2mm}\mbox{with probability of order 1,} \\
 1&\hspace*{-2mm}\mbox{with probability of order } T^{1/2}, \\
 >\hspace*{-1mm}1 &\hspace*{-2mm}\mbox{with probability smaller than } T.
 \end{cases}
\end{align*}
So that $\mathbb{V}\text{ar} \left( \log(\check{P}_{t+T}) - \log(\check{P}_t) \right)$ is of order $T^{1/2}$, as $T\searrow 0$, rather than of the usual order $T$. In the above sketchy argument we used the fact that, for every $t>0$, $$\lim_{h \rightarrow 0} \frac{1}{h} \ \mathbb{P} \left( \min_{n\in \mathbb{Z}} \vert W_t - n \vert \leq h \right) \geq c,$$ for some constant $c>0$.
\vspace{5mm}

To furnish a more precise result, we make - contrary to our usual assumption $W_0=0$ - the assumption that the Brownian motion starts from a random variable $W_0$ which is uniformly distributed on $[-1/2,1/2]$. Then, we can formulate the following more quantitative result for fixed $t=0$.

\begin{proposition}
Assume that $W_0$ is uniformly distributed on $[-1/2,+1/2]$. Then, 
$$ \liminf_{T \rightarrow 0} \sqrt{\frac{\pi}{2T}} \  \mathbb{V}\text{\emph{ar}} \left( \log(\check{P}_{T}) - \log(\check{P}_0) \right) =0.$$
\end{proposition}

\begin{proof}
Note that $$ \mathbb{V}\text{ar} \left( \log(\check{P}_{T}) - \log(\check{P}_0) \right) = \mathbb{E} \left[ \left( \log(\check{P}_{T}) - \log(\check{P}_0) \right)^2 \right], $$
where  $\log(\check{P}_{0})$ is in fact zero as we assumed that $W_0\sim$ Uni$(1/2,1/2)$. In the following $(B_t)_{t\geq 0}$ denotes a standard Brownian motion starting at $B_0=0$ such that $W_T = B_T + W_0$. Then,
\begin{align*}
\mathbb{E} &\left( \left( \log(\check{P}_{T}) - \log(\check{P}_0) \right)^2 \right)= \int_{-0.5}^{0.5} \mathbb{E}\left( \left( \text{int}(B_T + x) \right)^2 \right) dx \\
&=\int_{-0.5}^{0.5} \sum_{i=1}^\infty i^2 \left( \mathbb{P}\left(\frac{2i-1}{2}-x \leq B_T \leq \frac{2i+1}{2} - x \right) \right.\\
& \hspace*{25mm} \left.+  \mathbb{P}\left(-\frac{2i+1}{2}-x \leq B_T \leq -\frac{2i-1}{2} - x \right)\right) dx \\
&=\int_{-0.5}^{0.5} \sum_{i=1}^\infty i^2 \left( \Phi\left(\frac{i+0.5-x}{\sqrt{T}}\right) -  \Phi\left(\frac{i-0.5-x}{\sqrt{T}}\right) \right.\\
& \hspace*{25mm} \left.+  \Phi\left(\frac{i+0.5+x}{\sqrt{T}}\right) -  \Phi\left(\frac{i-0.5+x}{\sqrt{T}}\right) \right) dx  \\
&= \sum_{i=1}^\infty i^2  \left( \sqrt{\frac{2T}{\pi}} \left( \exp\left(-\frac{(i+1)^2}{2T}\right) + \exp\left(-\frac{(i-1)^2}{2T}\right) - 2 \exp\left(-\frac{i^2}{2T} \right) \right) \right. \\
& \hspace{25mm}+ \left. (2i+2) \Phi\left(\frac{i+1}{\sqrt{T}}\right) + (2i-2) \Phi\left(\frac{i-1}{\sqrt{T}}\right) - 4i  \Phi\left(\frac{i}{\sqrt{T}}\right) \right) \\
&= \sqrt{\frac{2 T}{\pi}} \left( 1 + 2 \sum_{i=1}^\infty \exp\left(-\frac{i^2}{2 T} \right) \right) - 4  \sum_{i=1}^\infty  i \Phi \left( -\frac{i}{\sqrt{T}} \right) \\
\end{align*}
We now use that fact for $x\rightarrow \infty$, $\Phi(-x) \approx \phi(x)/x$, where $\phi(x) = \exp(-x^2/2)/\sqrt{2\pi}$ is the probability density function of the standard normal distribution (we thank Friedrich Hubalek for pointing this out to us).
It follows that for small $T$
$$ i \Phi \left( -\frac{i}{\sqrt{T}} \right) \approx \sqrt{\frac{T}{2\pi}} \exp \left( -\frac{i^2}{2T} \right),
$$
which concludes the proof.
\end{proof}

\section*{Acknowledgements}

We thank Jean-Philippe Bouchaud, Rama Cont, Friedrich Hubalek and particularly Mathieu Rosenbaum for helpful comments as well as interesting discussions.
 
\bibliographystyle{abbrv}

\end{document}